\documentclass[a4paper,11pt,dvipsnames]{article}
 
\newcommand{\ignore}[1]{}
\usepackage{fullpage}
\usepackage[colorlinks=false]{hyperref}  
\usepackage{microtype}
\usepackage{textcomp}
\usepackage{amsmath, amsfonts, amssymb, amsthm}
\usepackage{mathtools}
\usepackage{mhsetup}
\usepackage{algorithm, algpseudocode, algorithmicx}
\usepackage{enumerate}

\DeclarePairedDelimiter\floor{\lfloor}{\rfloor}

\newtheorem{theorem}{Theorem}
\newtheorem{corollary}[theorem]{Corollary}
\newtheorem{lemma}[theorem]{Lemma}
\newtheorem{definition}[theorem]{Definition}

\newcommand{\claimproof}[2]%
{\noindent{\em Proof of Claim \ref{#1}.}
#2\hspace*{\fill}$\Box$~~~~~\vspace{5mm} }

\newenvironment{replemma}[1]{\medskip \par {\sc Lemma \ref{#1} (restated).} \em}{}
\newenvironment{reptheorem}[1]{\medskip \par {\sc Theorem \ref{#1} (restated).} \em}{}

\newcommand{\F}{\mathbb{F}}

\newcommand{\Q}{\mathbb{Q}}

\newcommand{\calC}{\mathcal{C}}
\newcommand{\calH}{\mathcal{H}}

\newcommand{\calT}{\mathcal{T}}
\newcommand{\calM}{\mathcal{M}}

\newcommand{\lrsp}{\text{sp}}
\newcommand{\rk}{\text{rk}}
\newcommand{\ch}{\text{char}}
\newcommand{\cf}{\text{coef}}
\newcommand{\poly}{\text{poly}}

\begin{document}

\title{\bf Small hitting-sets for tiny arithmetic circuits or: How to turn bad designs into good}

\author{
Manindra Agrawal \thanks{Department of Computer Science \& Engineering, IIT Kanpur, India, \texttt{manindra@cse.iitk.ac.in} }
\and 
Michael Forbes\thanks{Simons Institute for the Theory of Computing, University of California, Berkeley, USA, \texttt{miforbes@csail.mit.edu} }
\and
Sumanta Ghosh \thanks{CSE, IITK, \texttt{sumghosh@cse.iitk.ac.in} }
\and
Nitin Saxena \thanks{CSE, IITK, \texttt{nitin@cse.iitk.ac.in} }
}

\date{}
\maketitle

\begin{abstract}
Research in the last decade has shown that to prove lower bounds or to derandomize polynomial identity testing (PIT) for general arithmetic circuits it suffices to solve these questions for restricted circuits. In this work, we study the smallest possibly restricted class of circuits, in particular depth-$4$ circuits, which would yield such results for general circuits (that is, the complexity class VP). 
We show that if we can design poly($s$)-time hitting-sets for $\Sigma\wedge^a\Sigma\Pi^{O(\log s)}$ circuits of size $s$, where $a=\omega(1)$ is arbitrarily small and the number of variables, or arity $n$, is $O(\log s)$, then we can derandomize blackbox PIT for general circuits in quasipolynomial time. Further, this establishes that either E$\not\subseteq$\#P/poly or that VP$\ne$VNP. We call the former model \emph{tiny} diagonal depth-$4$. Note that these are merely polynomials with arity $O(\log s)$ and degree $\omega(\log s)$. 
In fact, we show that one only needs a poly($s$)-time hitting-set against individual-degree $a'=\omega(1)$ polynomials that are computable by a size-$s$ arity-$(\log s)$ $\Sigma\Pi\Sigma$ circuit (note: $\Pi$ fanin may be $s$).
Alternatively, we claim that, to understand VP one only needs to find hitting-sets, for depth-$3$, that have a small parameterized complexity. 

Another tiny family of interest is when we restrict the arity $n=\omega(1)$ to be arbitrarily small. In parameterized complexity terms: We show that if we can design poly($s,\mu(n)$)-time hitting-sets for size-$s$ arity-$n$ $\Sigma\Pi\Sigma\wedge$ circuits (resp.~$\Sigma\wedge^a\Sigma\Pi$), where function $\mu$ is arbitrary, then we can solve PIT for VP in quasipoly-time, and prove the corresponding lower bounds.

Our methods are strong enough to prove a surprising {\em arity reduction} for PIT-- to solve the general problem completely it suffices to find a blackbox PIT with time-complexity $sd2^{O(n)}$. This suggests that, in algebraic-geometry terms, PIT is inherently an `extremely low'-dimensional (or `extremely low' individual-degree) problem.

One expects that with this severe restriction on $n, a$ and the semantic individual-degree, it should be at least ``exponentially'' easier to design hitting-sets. Indeed, we give several examples of ($\log s$)-variate circuits where a new measure (called {\em cone-size}) helps in devising poly-time hitting-sets, but the same question for their $s$-variate versions is open till date: For eg., diagonal depth-$3$ circuits, and in general, models that have a {\em small} partial derivative space. The latter models are very well studied, following (Nisan \& Wigderson, FOCS'95 \cite{NW95}), but no $sd2^{O(n)}$-time PIT algorithm was known for them.

We also introduce a new concept, called {\em cone-closed basis} isolation, and provide example models where it occurs, or can be achieved by a small shift. This refines the previously studied notions of low-support (resp.~low-cone) rank concentration and least basis isolation in certain ABP models. Cone-closure holds special relevance in the low-arity regime.
\end{abstract}

\noindent
{\bf 1998 ACM Subject Classification:} F.1.1 Models of Computation, I.1.2 Algorithms, F.1.3 Complexity Measures and Classes 


\vspace{-1mm}
\section{Introduction}
\vspace{-1mm}

The Polynomial Identity Testing (PIT) problem is to decide whether a multivariate polynomial is zero, where the input is given as an arithmetic circuit. An arithmetic circuit over a field $\mathbb F$ is a layered acyclic directed graph with one sink node called output node; source nodes are called input nodes and are labeled by variables or field constants; non-input nodes are labeled by $\times$ (multiplication gate) and $+$ (addition gate) in alternate layers. Sometimes edges may be labeled by field constants. The computation is defined in a natural way. The complexity parameters of a circuit are: 1) \emph{size}- maximum number of edges and vertices, 2) \emph{depth}- maximum number of layers, and 3) \emph{degree}- maximum degree among all polynomials computed at each node. This is sometimes called the {\em syntactic} degree, to distinguish from the ({\em semantic}) degree of the final polynomial computed. The families of circuits, that are $n$-variate poly($n$)-size and poly($n$)-degree, define the class VP \cite{V79}; also see \cite{B13} for interesting variants of this algebraic computing model.

In this work we study $n$-variate polynomials computable by circuits of size $\le s$ of individual degree $\le a$, where one of the parameter is \emph{tiny} as compared to the others.  For example, we study such polynomials where the number of variables $n$ is very small, such as $n\le O(\log s)$. When $a\le O(1)$ this model is equivalent to that of $\poly(s)$-size depth-$2$ circuits, which are well-understood even in the blackbox model \cite{BT88}, and thus we probe these polynomials when $a\ge\omega(1)$.  We can even approach this question when $n$ is  $\omega(1)$ but arbitrarily small, in which case we now consider polynomials of individual degree $O(s)$ to again avoid collapsing to $\poly(s)$-size depth-$2$ circuits. Basically, we need to allow the number of monomials $a^n$ to grow as $s^{\omega(1)}$ for the model to be nontrivial, and we demonstrate that this is precisely the chasm to be crossed to get a fundamentally new understanding of VP.

The polynomial computed by a circuit may have, in the worst-case, an exponential number of monomials compared to its size. So, by computing the explicit polynomial from input circuit, we cannot solve PIT problem in polynomial time. However, evaluation of the polynomial at a point can be done, in time polynomial in the circuit size, by 
assigning the values at input nodes. This helps us to get a polynomial time randomized algorithm for PIT by evaluating the circuit at a random point, since any nonzero polynomial evaluated at a random point gives a nonzero value with high probability \cite{DL78, Zip79, Sch80}. However, finding a deterministic polynomial time algorithm for PIT is a long-standing open question in arithmetic complexity theory. It naturally appears in the algebraic approaches to the P$\ne$NP question, eg.~\cite{GMQ16, G15, mul12, Mul12b}. The famous algebraic analog is the VP$\ne$VNP question \cite{V79}.
The PIT problem has applications both in proving circuit lower bounds \cite{Kab03, Agrawal05} and in algorithm design \cite{Mul87, AKS04, Saraf14, Shpilka14}. For more details on PIT, see the surveys \cite{Saxena09, Saxena13, Shpilka10}. 

PIT algorithms are of two kinds: 1) \emph{whitebox}- allowed to see the internal structure of the circuit, and 2) \emph{blackbox}- only evaluation of the circuit is allowed at points in a small field extension. Blackbox PIT is equivalent to efficiently finding a set of points, called a \emph{hitting-set} $\calH$, such that for any circuit $C$, in a family $\calC$, computing a nonzero polynomial, the set $\calH$ must contain a point where $C\ne0$. For us a more functional approach would be convenient. We think in terms of an $n$-tuple of univariates $\mathbf f(y)=(f_1(y),\ldots, f_n(y))$ whose set of evaluations contain $\calH$. Such an $\mathbf f(y)$ can be efficiently obtained from a given $\calH$ (using interpolation) and vice-versa. Clearly, if $\calH$ is a hitting-set for $\calC$ then $C(\mathbf f(y))\ne0$, for any nonzero $C\in\calC$. This tuple of univariates is called a hitting-set generator (Sec.\ref{sec-tdd4}) and we will freely exploit this connection in all the proofs or theorem statements.

Existing deterministic algorithms solving PIT for restricted classes have been developed by leveraging insight into the weaknesses of these models. For example, deterministic PIT algorithms are known for subclasses of depth-$3$ circuits  \cite{KS07, Sax08, SS12}, subclasses of depth-$4$ circuits \cite{ASSS12, BMS13, SSS13, F15, KS16, KS16b, PSS16}, read-once arithmetic branching programs (ROABP) and related models \cite{FS12, Agrawal13, FSS14, AGKS15, GKST16, GKS16}, certain types of symbolic determinants \cite{FGT16, GT16}, as well as non-commutative models \cite{GGOW16}. An equally large number of special models have been used to prove lower bound results, see for example the ongoing survey of Saptharishi~\cite{Sap16}. 

While studying such restricted models may at first seem to give limited insight into general circuits, various works (discussed below) have shown this not to be the case as full derandomization of PIT for depth-$4$ (resp.~depth-$3$) circuits would imply derandomization of PIT for general circuits.  The goal of this work is to sharpen this connection by additionally limiting the \emph{number of variables} (resp.~semantic individual-degree) in the depth-$4$ circuit, and showing that such a connection still holds.  In doing so we establish new concepts for studying this small-variable regime, and show how to derive polynomial-size hitting sets for some small-variable circuit classes where only quasipolynomial-size, but not poly-sized, hitting-sets were previously known.

\vspace{-1mm}
\subsection{Main results}
\vspace{-1mm}

Arithmetic circuits were defined with the hope that they would have better structure than boolean circuits. Indeed, unlike boolean circuits, any VP circuit of arbitrary depth can be reduced nontrivially to depth-$4$ \cite{Agrawal08, K12, T15, CKSV16} or depth-$3$ \cite{Gupta13}. As a consequence, the lower bound questions against VP reduce to the lower bound questions for depth-$4$ (or to depth-$3$ for selected fields). In circuit complexity the base field $\F$ of interest is either $\Q$ or $\F_q$ (for a prime-$p$-power $q$). Other popular fields, eg.~number field, function field or $p$-adic field, are dealt with using similar computational methods. In this paper, unless stated otherwise, we assume $\F=\Q$. (Though many of our ideas would generalize to other base rings.) 

The PIT question for VP circuits reduces even more drastically. The reason is that now one invokes circuit factorization results \cite{K89} that use algebra in a way heavier than the depth-reduction results. So we will invoke that VP is closed under factorization, in addition to the fact that it affords depth-reduction. 
Recall the $\Sigma\Pi\Sigma\Pi$ (resp.~$\Sigma\wedge\Sigma\Pi$) model that computes a polynomial by summing products (resp.~{\em powers}) of sparse polynomials (see Defn.\ref{def:tiny_depth-4_diagonal_circuit}).
In 2008, Agrawal and Vinay \cite[Thm.3.2]{Agrawal08} showed that solving blackbox PIT in poly($s$)-time for size-$s$ $s$-variate depth-$4$ circuits of the form $\Sigma\Pi^a\Sigma\Pi^{O(\log s)}$ (Defn.\ref{def:tiny_depth-4_diagonal_circuit}), where $a$ is {\em any} unbounded function, gives an $(sd)^{O(\log sd)}$-time hitting-set for VP (size-$s$ degree-$d$). Here, we weaken the hypothesis further. We show that solving blackbox PIT in poly($s$)-time for size-$s$ {\em $O(\log s)$-variate} $\Sigma\wedge^a\Sigma\Pi^{O(\log s)}$ circuits, where $a$ is an arbitrarily small unbounded function and semantic individual-degree $a'$ is also arbitrarily small, is sufficient to get an $(sd)^{O(\log sd)}$-time hitting-set for VP. 
	We note that the brute-force deterministic algorithm would run here in time $a'^{O(\log s)}= s^{O(\log a')}$, and thus we show that reducing this runtime to polynomial would have dramatic consequences. We call such depth-$4$ circuits as \emph{tiny diagonal depth-$4$} (a better definition is Defn.\ref{def:tiny_depth-4_diagonal_circuit}). Compared to the previous result, one advantage in our model is that an exponential running time, wrt the number of variables ({\em arity} $n$), is allowed. Formally, we design an efficient arity reducing polynomial-map (the polynomials designed have individual-degree $O(1)$). Clearly, the map can be used to also deduce about the quasipoly-time blackbox PIT for VP.

\begin{theorem}\label{thm-main1}
Suppose we have poly-time hitting-sets for a tiny diagonal depth-$4$ model. Then, we design a poly($sd$)-time arity reducing ($n\mapsto O(\log sd)$) polynomial-map of constant individual-degree that preserves the nonzeroness of any $n$-variate size-$s$ degree-$d$ arithmetic circuit.
\end{theorem}

By the known depth-$3$ chasm \cite{Gupta13}, the hypothesis in Thm.\ref{thm-main1} can be weakened to: {\em if, for an arbitrary function $\mu'$, we have a poly($s,2^n,\mu'(a')$)-time hitting-set for size-$s$ arity-$n$ depth-$3$ circuits that compute polynomials of semantic individual-degree $\le a'$, then$\cdots$}. The proof sketch is given in Sec.\ref{app-0} (Thm.\ref{thm-main1.01}), where also the case of a tiny width-$2$ ABP (Thm.\ref{thm-main1.02}) and a `multilinear tiny' depth-$3$ variant (Thm.\ref{thm-main1.03}) are discussed. Note that the sparsity of a polynomial computed by tiny diagonal depth-$4$ is $a'^n = s^{O(\log a')}$ which gives us a brute-force hitting-set of similar complexity \cite{BT88}. We want to bring it down to $s^{O(1)}$; leaving us with an arbitrarily small gap to close algorithmically. Our methods show that any of these hitting-set designs will establish: Either E$\not\subseteq$\#P/poly or VNP has polynomials of arithmetic circuit complexity $2^{\Omega(n)}$ (Lem.\ref{lem-class-sep}, Cor.\ref{cor-hs-hard}). Note that these are long-standing open questions \cite{NW94, V79}. Their connection with PIT, in our results, is a significant strengthening of \cite{Kab03} who had first proved: PIT in poly-time implies that either NEXP$\not\subseteq$ P/poly or VNP$\ne$ VP.  

Moreover, we get the following curious property of PIT. (In some sense, it signifies: tiny-VP PIT implies VP PIT.)

\begin{theorem}[PIT arity reduction]\label{thm-main1.1}
If we have poly($sd2^n$)-time hitting-set for size-$s$ degree-$d$ arity-$n$ circuits, then for general circuits we have a poly($sd$)-time hitting-set (and we get an E-computable polynomial with exponential arithmetic complexity).
\end{theorem}

One now wonders whether the hypothesis, in the theorem above, can be further weakened. We give a partial answer by studying the model $\Sigma\Pi\Sigma\wedge$ (i.e.~a sum of products, where each factor is a sum of univariate polynomials).

\begin{theorem}[Tinier arity]\label{thm-main1.2}
Fix a function $\mu=\mu(s)$.
If we have poly($s, \mu(n)$)-time hitting-set for size-$s$ arity-$n$ $\Sigma\Pi\Sigma\wedge$ circuits, then for VP circuits we have a poly($sd$)-time arity reduction ($n\mapsto O(\log sd)$) that preserves nonzeroness (and proves an exponential lower bound).
\end{theorem}

The PIT algorithms in current literature always try to achieve a subexponential dependence on $n$, the number of variables. Our results demonstrate that all we need is a poly$(sd2^n)$-time algorithm to completely solve VP PIT. Or, a poly($s\mu(n)$)-time algorithm (for $\Sigma\Pi\Sigma\wedge$) to partially solve VP PIT and to prove ``either E$\not\subseteq$\#P/poly or VP$\ne$VNP''. For example, even a poly($s,A(n)$)-time hitting-set for $\Sigma\Pi\Sigma\wedge$, where $A$ is an Ackermann function \cite{A28}, would be tremendous progress. A similar case can be made for $\Sigma\wedge^a\Sigma\Pi(n)$ circuits, where both $a$ and $n$ are arbitrarily small unbounded functions, see Thm.\ref{thm-main1.3} (i.e.~time-complexity may be arbitrary in terms of both $a$ and $n$). 

Obviously, we should now discover techniques and measures that are specialized to this tiny regime. Many previous works use support size of a monomial as a measure to achieve rank concentration \cite{Agrawal13, FSS14, GKST16}. For a monomial $m$, its \emph{support} is the set of variables whose exponents are positive. We introduce a different measure: \emph{cone-size} (see Defn.\ref{def:cone}) which is the number of monomials that divide $m$ (also see \cite{Agrawal13, Forbes}). It has two advantages in the tiny regime. First, the number of monomials with cone-size at most $s$ is poly($s$) (Lem.\ref{lemma:cs_monomials}). Second, for any circuit $C$ and a monomial $m$, we devise (in \emph{blackbox}) a circuit $C'$ which computes the 
coefficient of $m$ in $C$ and has size polynomial in that of $C$ and the cone-size of $m$ (Lem.\ref{lemma:coef_extraction}). Using this measure we can define a new concept of rank concentration \cite{Agrawal13}-- `low'-cone concentration --and we are able to give poly-time hitting-sets for a large class of tiny circuits (i.e.~$n$ is logarithmic wrt size). We prove our result in a general form (Thm.\ref{thm-main2}) and as a corollary (Cor.\ref{cor:polyPIT_using_cs_for_tinyckt}) we get our claim. 
This gives us a poly-time hitting-set for depth-$3$ diagonal circuits where 
the rank of the linear forms is logarithmic wrt the size (Thm.\ref{thm:polyPIT_for_tiny_depth-3_diag_ckt}). 

\begin{theorem}\label{thm-main2}
Let $\mathcal C$ be a set of arity-$n$ degree-$d$ circuits with size-$s$ s.t.~for all $C\in \mathcal C$, the dimension of the partial derivative space of $C$ is at most $k$. Then, blackbox PIT for $\mathcal C$ can be solved in $(sdk)^{O(1)}\cdot (3n/\log k)^{O(\log k)}$ time.
\end{theorem}

Note that for $n=O(\log k)=$ $O(\log sd)$, the above bound is poly-time and such a PIT result was not known before. For instance, general diagonal depth-$3$ is a prominent model with low partial derivative space; it has a whitebox poly-time PIT \cite{Sax08} but no poly-time hitting-set is known (though \cite{FSS14} gave an $s^{O(\log\log s)}$-time hitting-set.). Even for $O(\log k)$-variate diagonal depth-$3$ no poly-time hitting-sets were known before our work.

We investigate another structural property useful in the tiny regime. Consider a polynomial $f(\mathbf x)$ with coefficients over $\F^k$. Let $\lrsp(f)$ be the subspace spanned by its coefficients. We say that $f$ has a {\em cone-closed basis} if there is a set of monomials $B$ whose coefficients form a basis of $\lrsp(f)$ and if $B$ is closed under submonomials. We prove that this notion is a strengthening of both low-support \cite{Agrawal13} and low-cone concentration ideas \cite{F15} (see Lem.\ref{lemma:cone-closed_cs}). Recently, this notion of closure has also appeared as an {\em abstract simplicial complex} in \cite{GKPT16}.

Interestingly, we show that a general polynomial $f$, when shifted by a `random' amount (or by `formal' variables), becomes cone-closed. More generally, we prove the following theorem relating this concept to that of basis isolating weight assignment \cite{AGKS15}.

\begin{theorem}\label{thm-main3}
Let $f(\mathbf x)\in\mathbb F[\mathbf x]^k$ be an arity-$n$ degree-$d$ polynomial over $\F^k$. Let $\mathbf w$ be a basis isolating weight assignment of $f(\mathbf x)$. Then $f(\mathbf x+t^{\mathbf w})$ has a cone-closed basis over $\F(t)$.
\end{theorem}

\vspace{-1mm}
\subsection{Proof ideas}
\vspace{-1mm}

The proof of Thm.\ref{thm-main1} is a technical refinement of the strategy of \cite[Thm.3.2]{Agrawal08} in at least three ways. 
First, we get a polynomial $g$ that is hard for tiny diagonal depth-$4$ circuits $\calC$ from the hitting-set for $\calC$ (Lem.\ref{lemma:PIT_to_lb}). Essentially, $g$ will be an {\em annihilating polynomial} of the hitting-set generator. The novelty here is that we have to allow $g$ to be {\em non}-multilinear (unlike \cite[Lem.3.3]{Agrawal08}), for it to exist, as the arity of $\calC$ in our case is logarithmically small. We show that a {\em multi-$\delta$-ic} $g$ suffices (i.e.~$g$ has individual degree bounded by a constant $\delta$). By VP depth-reduction, and Fischer's trick (that is `cheap' to apply in the tiny regime), this polynomial remains hard for VP. Next by Lem.\ref{lemma:lb_to_PIT}, where we use Nisan-Wigderson design \cite{NW94} and Kaltofen's factorization \cite{K89}, we get a poly-time arity-reducing polynomial-map based on $g$ that keeps any nonzero VP circuit nonzero and reduces its arity from $n$ to $O(\log sd)$. This gives the theorem. 

The `$\log$' function in our method comes from the use of Nisan-Wigderson design and because we want the $\delta$ to be constant in depth-reduction; reducing the arity further would need a new idea. Moreover, we observe that the annihilator $g$ is an E-computable polynomial with exponential arithmetic circuit complexity. This means that: Either E$\not\subseteq$\#P/poly or VNP has a polynomial with exponential arithmetic complexity (Lem.\ref{lem-class-sep}). It is not clear whether we can strengthen the connection all the way to the (conjectured) VNP$\ne$VP. Perhaps, this will require starting with a more structured hitting-set generator for $\calC$, so that its annihilator $g$ is a polynomial whose coefficient bits are (\#P/poly)-computable (see Valiant's criterion \cite[Prop.2.20]{B13}).

The proof of Thm.\ref{thm-main1.2} requires one to move in a different regime where the arity is $n=a\log s = \omega(\log s)$, the semantic individual-degree is {\em one} and the circuit is depth-$3$ (Thm.\ref{thm-main1.03}). We reach there in a sequence of steps, and then apply the Kronecker map ($x_i\mapsto y^{2^i}$) {\em locally} in blocks of size $\log s$. This leads to an arity reduction $a\log s\mapsto a$ in the PIT problem. 
 
The proof of Thm.\ref{thm-main2} has two steps. In the first step, we show that with respect to any \emph{monomial ordering}, the dimension $k$ of the partial derivative space of a polynomial is lower bounded by the cone-size of its \emph{leading monomial}. So, for every nonzero $C\in \mathcal C$ there is a monomial with nonzero coefficient and cone-size $\le k$. The second step is to check whether the coefficients of all the monomials in $C$, with cone-size $\leq k$, are zero. Lem.\ref{lemma:coef_extraction} describes the time required to check whether the coefficient of a monomial is zero. Lem.\ref{lemma:cs_monomials} gives us an optimal upper bound on the number of monomials with cone-size $\leq k$. 

Thm.\ref{thm-main3} unfolds an interesting combinatorial interaction between variable shift and cones of resulting monomials. Let a set of monomials $B$ be the least basis, wrt to the basis isolating weight assignment, of $f$. We consider the set of all submonomials of those in $B$ and identify a subset $A$ that is cone-closed. We define $A$ in an algorithmic way, as described in Algo.\ref{algo:find_con-closed_set}. The fact that $A$ is exactly a basis of the shifted $f$ is proved in Lem.\ref{lemma:cone-closed_nvar_basis} by studying the action of the shift on the coefficient vectors. This has an immediate (nontrivial) consequence that any polynomial $f$ over $\F^k$, when shifted by formal variables, becomes cone-closed.

\vspace{-1mm}
\section{Tiny diagonal depth-$4$ circuits-- Proof of Theorems \ref{thm-main1}-\ref{thm-main1.2} }\label{sec-tdd4}
\vspace{-1mm}

In this section we will revisit the techniques that have appeared in some form in \cite{NW94, Kab03, Agrawal05, Agrawal08, Gupta13} and strengthen them to derive our results. First, we show how to get a hard polynomial from a hitting-set.

\smallskip\noindent {\bf Hitting-set generator.}
Let $\mathcal{C}$ be a set of arity-$n$ circuits. We call an $n$-tuple of univariates $\mathbf f(y)=(f_1(y),\ldots, f_n(y))$ a {\em $(t, d)$-hsg} (hitting-set generator) for $\mathcal C$ if: (1) for any nonzero $C\in\calC$, $C(\mathbf f(y))\ne0$,  and (2) $\mathbf f$ has time-complexity $t$ and the degree of each $f_i$ is at most $d$. 

From this, we get a hard polynomial simply by looking at an annihilating polynomial of $\mathbf f(y)$.

\begin{lemma}[Hitting-set to hardness]
\label{lemma:PIT_to_lb}
Let $\mathbf f(y)=(f_1(y),\ldots, f_n(y))$ be a $(t, d)$-hsg for $\mathcal C$. Then, there exists an arity-$n$ polynomial $g(\mathbf x)$ that is not in $\mathcal C$, is computable in $\mathrm{poly}(tdn)$-time, and has individual 
degree less than $2\delta := 2\lceil (dn+1)^{1/(n-1)} \rceil$. Moreover, we can ensure that its degree is exactly $\delta n$. 
\end{lemma}

\begin{proof}
A natural candidate for $g(\mathbf x)$ is any annihilating polynomial of the $n$ polynomials $\mathbf f(y)=(f_1(y),\ldots, f_n(y))$, since for every nonzero $h\in \mathcal C$, $h(\mathbf f)$ is nonzero. Define $\delta$ as the smallest integer such that $\delta^{n-1} > dn\ge (\delta-1)^{n-1}$. Consider $g(\mathbf x)$ as an arity-$n$ polynomial with individual degree less than $\delta$. 
Then, $g(\mathbf x)$ can be written as: 
\begin{equation}\label{eqn-ann-poly}
g(\mathbf x) \,=\, \sum_{\mathbf e\colon 0\le \mathbf e_i<\delta}c_{\mathbf e}\mathbf{x^e} 
\end{equation}
where, $c_{\mathbf e}$'s are unknown to us. We can set all these $c_{\mathbf e}$'s to zero except the ones corresponding to an index-set $I$ of size $\delta_0:=(dn\delta+1)$. This way we have only $\delta_0$ unknowns. To be an annihilating polynomial of $\mathbf f(y)$, we need $g(\mathbf f)=0$. By comparing the coefficients of the monomials, in $y$, both sides we get a linear system in the unknowns. 

Suppose that $\delta_1$ is the degree of $y$ in $g(\mathbf f)$. Then, $g(\mathbf f)$ can be written as $g(\mathbf f)=\sum_{i=0}^{\delta_1} p_i\cdot y^i$, where $p_i$'s are linear polynomials in $c_{\mathbf e}$'s. The constraint $g(\mathbf f)=0$ gives us a system of linear equations with the number of unknowns $\delta_0$ and the number of equations $\delta_1$. The value of $\delta_1$ can be at most 
$d\cdot n\cdot\delta$, which means that the number of unknowns $\delta_0$ is greater than the number of equations $\delta_1$. So, our system of linear equations always has a nontrivial solution, which gives us a nonzero $g$. In case its degree is $<\delta n$, we can multiply by an appropriate monomial $\mathbf{x^e}$ to make the degree $=\delta n$. Note that $g\notin\mathcal C$ holds, as $g(\mathbf f)$ is still nonzero (since $f_1(y),\ldots, f_n(y)$ are nonzero wlog).

Computing $\mathbf f(y)$ takes $t$ time and a solution of the linear equations can be computed in $\mathrm{poly}(tdn)$-time. So, $g(\mathbf x)$ can be computed in $\mathrm{poly}(tdn)$-time. 
\end{proof}

\begin{corollary}[E-computable]\label{cor-hs-hard}
In the proof of Lem.\ref{lemma:PIT_to_lb}, if $td=2^{O(n)}$ then the polynomial family $g_n :=g$, indexed by the arity, is E-computable (i.e.~all the formal monomials and their coefficients can be produced in poly($2^n$)-time given unary $n$).
\end{corollary}
\begin{proof}
The linear system that we got can be solved in poly($tdn$)-time. As it is homogeneous we can even get an integral solution in the same time-complexity. Thus, assuming $td=2^{O(n)}$, the time-complexity of computing a bit of $\text{coef}_{\mathbf{x^e}}(g)$ is poly($tdn$)$=$poly($2^n$), the coefficient bitsize is poly($2^n$) and so is the number of monomials in $g$ ($\because \delta=O(1)$). In other words, if we consider the polynomials $g_n := g$, indexed by the arity, then the family $\{g_n\}_n$ is E-computable.
\end{proof}

Towards a converse of the above lemma, a crucial ingredient is the Nisan-Wigderson design \cite{NW94}.

\begin{definition}\label{def:(l,n,d)-design}
Let $\ell>n>d$. A family of subsets $\mathcal F=\{I_1,\ldots, I_m\}$ on  $[\ell]$ is called an {\em $(\ell,n,d)$-design}, if $|I_i|=n$ and for all $i\neq j\in[m]$, $|I_i\cap I_j|\leq d$. 
\end{definition}

\begin{lemma}[Nisan-Wigderson design, Chap.16 \cite{AB09}]
\label{lemma:NW_design}
There exists an algorithm which takes $(\ell, n, d)$ and a base set $S$ of size $\ell> \frac{10n^2}{d}$ as input, and outputs an $(\ell, n, d)$-design $\mathcal F$ having $\geq 2^{\frac{d}{10}}$ subsets, in time $2^{O(\ell)}$.
\end{lemma}

Next, we use a hard polynomial $q_m$ on a small design to get a poly-time computable arity-reducing polynomial-map for VP that preserves nonzeroness.

\begin{lemma}[Hardness to VP reduction]
\label{lemma:lb_to_PIT}
Let $\{q_m\}_{m\geq 1}$ be a family of multi-$(\delta-1)$-ic polynomials such that 
it can be computed in $\delta^{O(m)}$ time, but has no $\delta^{o(m)}$-size arithmetic circuit. Then there is a $\delta^{O(\log sd)}$-time arity reduction, from $n$ to $O(\log sd)$, for VP circuits. 
\end{lemma}

\begin{proof}
Note that there is a constant $c_0>0$ such that $q_m$ requires $\Omega(\delta^{c_0m})$-size arithmetic circuits. Otherwise $\{q_m\}_{m\geq 1}$ will be in $\cap_{c>0} \text{Size}(\delta^{cm})$, and hence in $\text{Size}(\delta^{o(m)})$.

Let $\mathcal C$ be a set of arity-$n$ degree-$d$ VP circuits with size$\leq s$. Let $n':=sd \ge n$. Let $\mathcal F=\{S_1,\ldots, S_{n'}\}$ be a $(c_2\log n', c_1\log n', 10\log n')$-design on the variable set $Z=\{z_1,\ldots, z_{c_2\log n'}\}$. Constants $c_2>c_1>10$ will be fixed later. Our hitting-set generator for $\mathcal C$ is defined as: for all 
$i\in[n]$, $x_{i}=q_{c_1\log n'}(S_i) =:p_i$. Then, we show that for any nonzero polynomial $C(\mathbf x)\in \mathcal C$, $C(p_1,\ldots, p_n)$ is also nonzero. 

For the sake of contradiction, assume that $C(p_1,\ldots, p_n)$ is zero. Since $C(\mathbf x)$ is nonzero, we can find the smallest $j\in[n]$ such that 
$C(p_1,\ldots, p_{j-1}, x_j,\ldots, x_n) =:C_1$ is nonzero, but 
$C_1(x_j=p_j)$ is zero. Thus, $(x_j-p_j)$ divides  $C_1$. Let $\mathbf a$ be an assignment on all the variables in $C_1$, except $x_j$ and the variables $S_j$ in $p_j$, with the property: $C_1$ at $\mathbf a$ is nonzero. Since $C_1$ is nonzero, we can find such an assignment. Now our new polynomial $C_2$ on the variables $S_j$ is of the form: 
$$C_2(S_j) \,=\, C(p_1',\ldots, p_{j-1}',x_j, a_{j+1},\ldots, a_n) $$ 
where, for each $i\in[j-1]$, $p_i'$ is the polynomial on the variables $S_i\cap S_j$, and $a_i$'s are field constants decided by our assignment $\mathbf a$. By the design, for each $i\in[j-1]$, $|S_i\cap S_j|\leq 10\log n'$. Since $p_i'$ are polynomials on variables $S_i\cap S_j$ of individual degree$< \delta$, each $p_i'$ has a circuit of size at most $n\delta^{10\log n'}\le\delta^{11\log n'}$. Then we have a circuit for $C_2$ of size at most $s_1:= s+ n\cdot\delta^{11\log n'}$, and degree at most $d_1:= d\cdot \delta c_1\log n'$. Since $(x_j-p_j)$ divides $C_2$, we can invoke the VP factorization algorithm \cite{K89} (see \cite[Thm.2.21]{B13} for the arithmetic circuit complexity of factors) and get an arithmetic circuit for $p_j$ of size $(s_1d_1)^{c_3}$, for some absolute constant $c_3$ (independent of $c_1, c_2$). 

Now we fix constants $c_1, c_2$. Pick $c_1$ such that $\delta^{c_0\cdot c_1\log n'}$ is asympotically larger than $(2sn\delta^{11\log n'}\cdot d\delta c_1\log n')^{c_3} > (s_1d_1)^{c_3}$. Since $sd=n'$ and $\delta\ge2$, the absolute constant $c_1:= 15c_3/c_0$ (independent of $c_2$) satisfies the above condition. 

Pick $c_2$, following Lem.\ref{lemma:NW_design}, such that $c_2\log n' > 10\cdot(c_1\log n')^2/(10\log n')$. So, $c_2:= 1+c_1^2$ works. With these values of $c_1, c_2$, we have a design that `stretches' $c_2\log n'$ variables to $n$ subsets with the required `low' intersection property. It is computable in poly($n'$)-time.

Moreover, if $C(p_1,\ldots, p_n)$ is zero then, by the above discussion, $p_j=q_{c_1\log n'}(S_j)$ has a circuit of size $(s_1d_1)^{c_3}= o(\delta^{c_0\cdot c_1\log n'})$. This violates the lower bound hypothesis. Thus, $C(p_1,\ldots, p_n)$ is nonzero. 

The time for computing $(p_1,\ldots, p_n)$ depends on: (1) computing the design (i.e.~poly($n'$)-time), and (2) computing $q_{c_1\log n'}$ (i.e.~$\delta^{O(\log n')}$-time). Thus, the arity reduction map for VP is computable in $\delta^{O(\log n')}$ time.
\end{proof}

Once we have a polynomial that is hard for a tiny model, to apply the above lemma, we need to show that it is also hard for VP. This is done by depth-reduction results.  First, we need a lemma that converts a monomial into a sum of powers. This was used in \cite{Gupta13}. (It requires $\ch\;\F=0$ or large.)

\begin{lemma}[Fischer's Trick \cite{F94}]
\label{lemma:fischer's_trick}
Over a field $\mathbb F$ of $\ch(\mathbb F)=0\ \mathrm{or}\ >r$, any expression of the form $g=\sum_{i\in [k]}\prod_{j\in[r]} g_{ij}$ with $\deg(g_{ij})\leq \delta$, can be rewritten as $g=\sum_{i\in[k']} c_i g_i^r$ where $k':= k2^r$, $\deg(g_i)\leq \delta$ and $c_i\in \mathbb F$.
\end{lemma}

Motivated by this transformation (when $2^r$ is `small'), we define a tiny subclass of VP.

\begin{definition}
\label{def:tiny_depth-4_diagonal_circuit}
The {\em diagonal depth-$4$} circuits compute polynomials of the form $\sum_{i\in[k]} c_i f_i^a$ where $f_i$'s are sparse polynomials in $\mathbb F[x_1,\ldots, x_n]$ of degree $\leq b$ and $c_i$'s in $\F$. A standard notation to denote this class is $\Sigma\wedge^a\Sigma\Pi^b(n)$. This is a special case of the depth-$4$ $\Sigma\Pi^a\Sigma\Pi^b(n)$ model that computes polynomials of the form $\sum_{i\in[k]} \prod_{j\in[a]} f_{i,j}$ where $f_{i,j}$'s are sparse polynomials in $\mathbb F[x_1,\ldots, x_n]$ of degree $\leq b$.

Given a constant $c>1$, computable functions $\mu(a)=\Omega(a)$ and $\mu'$, we define the class $\calT_{\mu, \mu', c}$ , called \emph{tiny diagonal depth-$4$}, containing $\Sigma\wedge^a\Sigma\Pi^b(n)$ circuits of size $\leq s$ that compute polynomials of semantic individual-degree $\le a'$ and $2^n + 2^b + \mu(a)+ \mu'(a') \;<$ $s^c$.

Analogously, we define the class $\calT_{\mu',c}'$ , called \emph{tiny depth-$3$}, containing $\Sigma\Pi\Sigma(n)$ circuits of size $\leq s$ that compute polynomials of semantic individual-degree $\le a'$, and $2^n + \mu'(a') \;<$ $s^c$.
\end{definition}

{\bf Remark.}
Note that $n, b=$ $O(\log s)$ and by picking the function $\mu(\cdot)$ (resp.~$\mu'$) arbitrarily large we can make $a=a(s)=\omega(1)$ (resp.~$a'$) an arbitrarily small computable function. Also, in this regime the number of monomials in the bottom  $\Sigma\Pi^b$ layer is ${n+b\choose b}<2^{n+b}< s^{2c}$, so the $f_i$'s can be thought of as given in the {\em dense} representation. Analogously, in tiny depth-$3$, $n=O(\log s)$ and the semantic individual-degree bound $a'=\omega(1)$ can be picked arbitrarily small.

An alternative interpretation of the tiny models can be given using the parameterized complexity \cite{DF13} of PIT. Essentially, we are interested in hitting-sets for the diagonal depth-$4$ model that are {\em fixed parameter tractable} wrt $n,b, a$ and $a'$ (input size is $s$). Analogously, we are interested in hitting-sets for the depth-$3$ model that are {\em fixed parameter tractable} wrt $n$ and $a'$, where $a'$ is the semantic individual-degree bound (input size is $s$).

Now we invoke VP depth-reduction to get to the tiny model, and finish our proof.

\begin{proof}[Proof of Thm.\ref{thm-main1}]
The proof is along the lines of \cite[Thm.3.2]{Agrawal08}. Using Lem.\ref{lemma:PIT_to_lb}, from poly-time hitting-set generator for tiny diagonal depth-$4$, we get a hard polynomial for this model. Then we show that it is also hard for VP and invoke Lem.\ref{lemma:lb_to_PIT} to get the VP arity reduction. 

Now we provide the details. Let constant $c>1$, functions $\mu(a)=\Omega(a)$ and $\mu'$ be given in the hypothesis. 
Let $\mathcal C\subset \calT_{\mu,\mu', c}$ be the set of tiny diagonal depth-$4$ circuits of size $\le s$ and arity $m:=\log s$. Assume that $\mathcal C$ has a $(s^e, s^e)$-hsg $\mathbf f(y)$ for some constant $e\ge1$. Then using Lem.\ref{lemma:PIT_to_lb}, we have an $m$-variate polynomial $q_m$ with individual degree less than some constant $2\delta$, $\delta^{m-1}=\delta^{\log s-1}> s^e \log s$, and $q_m$ is computable in $s^{O(1)}= \delta^{O(m)}$ time. It has degree $=\delta m$. Importantly, $q_m\notin \calT_{\mu,\mu', c}$ , thus, no tiny diagonal depth-$4$ circuit of size $\leq s = \delta^{\Theta(m)}$ can compute it (otherwise, $q_m(\mathbf f(y))\ne 0$ which contradicts its definition as an annihilator). Next we show that it is also not computable by any $\delta^{o(m)}$-size arithmetic circuit. 

For the sake of contradiction, assume that $q_m$ has a $\delta^{o(m)}$-size circuit. From depth-reduction results \cite{Sap16} we get a circuit $C$, of $\Theta(\log \delta m)$-depth and $s_m= \delta^{o(m)}$ size, with the additional properties: 
\begin{enumerate}
\item alternative layers of addition/multiplication gates with the top-gate (root) being addition. 
\item below each multiplication layer the polynomial degree at least halves.
\item fan-in of each multiplication gate is at most $5$. 
\end{enumerate}

Now we cut the circuit $C$ at the $t$-th layer of multiplication gates from the top, where $t=t(s_m)$ will be fixed later, to get the two parts: 
\begin{description}
\item[Top part:] the top part computes a polynomial of degree at most $5^t$ and the number of variables 
is at most $s_m$. So it can be reduced to a $\Sigma\Pi$ circuit of size $\binom{s_m+5^t}{5^t}=s_m^{O(5^t)}$ (Stirling's approximation, see \cite[Prop.4.4]{Sap16}).

\item[Bottom part:] in the bottom part, we can have at most $s_m$ many top-multiplication gates that feed into the top part as input. Each multiplication gate computes a polynomial of degree at most $\delta m/2^t$ and the number of variables is at most $m$. So each multiplication gate can be reduced to a $\Sigma\Pi$ circuit of size $\binom{m+ \delta m/2^t}{\delta m/2^t} =$ $2^{O(\delta mt/2^t)}$. 
\end{description}

From the above discussion, we have a $\Sigma\Pi^{5^t}\Sigma\Pi^{\delta m/2^{t}}$ circuit $C'$, computing $q_m$, that has size $s_m^{O(5^t)}+s_m\cdot 2^{O(\delta mt/2^t)}$. 

The second summand becomes $2^{o(m\log \delta)}$ if we pick $t=\omega(1)$ (recall that $s_m=\delta^{o(m)}$ and $\delta=O(1)$). To get a similar upper bound on the first summand we need to pick $5^t\log s_m = o(m\log \delta)$. Finally, we also want $\mu(5^t)=o(m)$. A function $t=t(s_m)=t(s)$, satisfying the three conditions, exists as $\log s_m = o(m\log \delta)$ and $\mu(\cdot)$ is an increasing function. Let us fix such a function $t$. (As $C$ has super-constant depth, we can also assume that the cut at depth $t$ will be possible.) Thus the circuit $C'$, computing $q_m$, has size $s'_m=$ $\delta^{o(m)}$.

Let $a:=5^t$ and $b:=\delta m/2^{t}$. Consider the measure $E:= 2^m+ 2^b+ \mu(a)+\mu'(2\delta)$. We have the estimate $E=$ $s+2^{o(m)}+o(m)+O(1) =$ $s+o(s)=$ $o(s^c)$. 
So now we have a shallow circuit for $q_m$ of the form $\Sigma\Pi^{a}\Sigma\Pi^{b}$. Applying Lem.\ref{lemma:fischer's_trick}, we get a tiny diagonal depth-$4$ circuit, in $\calT_{\mu,\mu',c}$, computing $q_m$ of the form $\Sigma\wedge^{a}\Sigma\Pi^{b}$ and size $s'_m\cdot 2^a = \delta^{o(m)}\cdot 2^{O(\mu(a))} = \delta^{o(m)}$ which is $<s$. This contradicts the hardness of $q_m$. Thus, there is no arithmetic circuit for $q_m$ of size $\delta^{o(m)}$. 

Now invoking Lem.\ref{lemma:lb_to_PIT} on the hard family $\{q_m\}_{m\ge1}$, we get our claim.
\end{proof}

\begin{proof}[Proof of Thm.\ref{thm-main1.1}]
Suppose we have a poly($sd2^n$)-time hitting-set $\calH_{s,d,n}$ for size-$s$ degree-$d$ arity-$n$ circuits. Then, in particular, we have a poly($s$)-time blackbox PIT for tiny diagonal depth-$4$. Thus, Thm.\ref{thm-main1} gives a poly($sd$)-time arity-reducing polynomial-map ($n\mapsto O(\log sd)$) for VP that preserves nonzeroness. Let $n':=sd$ correspond to a given VP circuit $C$. Now, using $\calH_{n'^{O(1)},O(d\log n'),O(\log n')}$ we get a poly($n'$)-time hitting-set for $C$.

In the proof of Thm.\ref{thm-main1} we get the hard polynomial $q_m$, which by Cor.\ref{cor-hs-hard} gives us an E-computable polynomial family $\{q_m\}_m$, indexed by the arity, that has arithmetic circuit complexity $2^{\Omega(m)}$. 
\end{proof}

The existence of such a family $\{q_m\}_m$ has interesting complexity consequences.

\begin{lemma}[Class separation]\label{lem-class-sep}
If we have an E-computable polynomial family $\{f_n\}_n$ with individual-degree $O(1)$ and arithmetic circuit complexity $2^{\Omega(n)}$, then either E$\not\subseteq$\#P/poly or VNP has polynomials of arithmetic complexity $2^{\Omega(n)}$.
\end{lemma}
\begin{proof}
Say, for a constant $\delta\ge1$, we have an E-computable multi-$\delta$-ic polynomial family $\{f_n\}_n$ with arithmetic circuit complexity $2^{\Omega(n)}$. Clearly, the coefficients in $f_n$ have bitsize $2^{\Omega(n)}$. By using a simple transformation, given in \cite[Lem.3.9]{KP09}, we get a multi-$\delta$-ic polynomial family $\{h_n\}_n$, that is E-computable and has arithmetic complexity $2^{\Omega(n)}$, such that its coefficients are $\{0, \pm1\}$.

Assume E$\subseteq$\#P/poly. Since each coefficient of $h_n$ is a signed-bit that is computable in E, we deduce that the coefficient-function of $h_n$ is in \#P/poly. Thus, by \cite[Prop.2.20]{B13}, $\{h_n\}_n$ is in VNP and has arithmetic complexity $2^{\Omega(n)}$.
\end{proof}

Our techniques could  handle many other `tiny' models. The proofs are given in Sec.\ref{app-0}.

\begin{theorem}[Tiny depth-$3$]\label{thm-main1.01}
If we have poly-time hitting-sets for a tiny depth-$3$ model, then for VP circuits we have a poly($sd$)-time arity reduction ($n\mapsto O(\log sd)$) that preserves nonzeroness (and proves an exponential lower bound).
\end{theorem}

\begin{theorem}[Width-$2$ ABP]\label{thm-main1.02}
If we have poly($s2^n$)-time hitting-sets for size-$s$ arity-$n$ width-$2$ upper-triangular ABP, then for VP circuits we have a poly($sd$)-time arity reduction ($n\mapsto O(\log sd)$) that preserves nonzeroness (and proves an exponential lower bound).
\end{theorem}

Our method could also handle individual-degree $a'=O(1)$ (eg.~multilinear polynomials), but then we have to allow arity $\omega(\log s)$ (clearly, arity $O(\log s)$ trivializes the model \cite{BT88}). We state our result below in  parameterized complexity terms. (Proof in Sec.\ref{app-0}.)

\smallskip\noindent
{\bf Multilinear tiny depth-$3$.}
Given a constant $c>1$ and an arbitrary function $\mu'(a')=\Omega(a')$, we define the class $\calM_{\mu',c}$ , called multilinear tiny depth-$3$, containing $\Sigma\Pi\Sigma(n)$ circuits of size $\leq s$ that compute multilinear polynomials, and $\mu'(n/\log s) \;<$ $s^c$.

\begin{theorem}[Multilinear tiny depth-$3$]\label{thm-main1.03}
If we have poly-time hitting-sets for a multilinear tiny depth-$3$ model, then for VP circuits we have a poly($sd$)-time arity reduction ($n\mapsto O(\log sd)$) that preserves nonzeroness (and proves an exponential lower bound).
\end{theorem}

\subsection{Arbitrarily small arity suffices}

Using the previous result we can now reduce the arity, for PIT purposes, arbitrarily.

\begin{proof}[Proof of Thm.\ref{thm-main1.2}]
Suppose we have a poly($s,\mu(n)$)-time hitting-set $\calH_{s,n}$ for size-$s$ arity-$n$ $\Sigma\Pi\Sigma\wedge$ circuits. Wlog we can assume that $\mu(n)= \Omega(n)$. Let $a=a(s)$ be a function satisfying $\mu(a)\le s$, and define $n=n(s):= a\log s$. Consider a size-$s$ $\Sigma\Pi\Sigma(n)$ circuit $C\ne0$ computing a multilinear polynomial. We intend to design a hitting-set for $C$.

Partition the variable set $\{x_1,\ldots,x_n\}$ into $a$ blocks $B_j, j\in[a]$, each of size $\log s$. Let $B_j= \{x_{u(j)+1}, x_{u(j)+2},\ldots, x_{u(j)+\log s}\}$, for all $j\in[a]$ (pick $u$ to be an appropriate function). Consider the arity-reducing ``local Kronecker'' map $\varphi: x_{u(j)+i} \mapsto y_j^{2^i}$. Note that $\varphi(C) \in \F[y_1,\ldots,y_a]$, and its semantic individual-degree is at most $2s$. 

It is easy to see that $\varphi(C)\ne0$ (basically, use the fact that $C$ computes a nonzero multilinear polynomial and $\varphi$ keeps the multilinear monomials distinct). Finally, $\varphi(C)$ becomes an arity-$a$ $\Sigma\Pi\Sigma\wedge$ circuit of size at most $s + s\cdot 2^{\log s}= O(s^2) $. 
Thus, using $\calH_{O(s^2), a}$ we get a hitting-set for $\varphi(C)$ of time-complexity poly($s^2,\mu(a)$)= poly($s$). In turn, we get a poly-time hitting-set for multilinear tiny depth-$3$ model $\calM_{\mu,2}$. By invoking Thm.\ref{thm-main1.03} we finish the argument.
\end{proof}

We can also work with a version of diagonal depth-$4$ with arbitrarily small $n$ \& $a$.

\begin{theorem}[Tinier $n,a$]\label{thm-main1.3}
Fix functions $\mu=\mu(s)$ and $\mu'$.
If we have poly($s, \mu(a), \mu'(n)$)-time hitting-set for size-$s$ $\Sigma\wedge^a\Sigma\Pi(n)$ circuits, then for VP circuits we have a poly($sd$)-time arity reduction ($n\mapsto O(\log sd)$) that preserves nonzeroness (and proves an exponential lower bound).

{\em (Proved in Sec.\ref{app-0}.)}
\end{theorem}

\vspace{-1mm}
\section{Low-cone concentration and hitting-sets-- Proof of Thm.\ref{thm-main2}}\label{sec-lcc}
\vspace{-1mm}

In this section we initiate a study of properties that are relevant for tiny circuits (or the log-arity regime). 

\begin{definition}[Cone of a monomial]
\label{def:cone}
A monomial $\mathbf{x^e}$ is called a \emph{submonomial} of $\mathbf{x^f}$, 
if $\mathbf e\leq \mathbf f$ (i.e.~coordinate-wise). We say that $\mathbf{x^e}$ is a \emph{proper submonomial} of $\mathbf{x^f}$, if $\mathbf e\leq \mathbf f$ and $\mathbf{e\neq f}$.

For a monomial $\mathbf{x^e}$, the \emph{cone of} $\mathbf{x^e}$ is the set of all submonomials of $\mathbf{x^e}$. The cardinality of  this set is called \emph{cone-size of} $\mathbf{x^e}$. It equals $\prod (\mathbf{e+1}):= \prod_{i\in[n]} (e_i+1)$, where $\mathbf e=(e_1,\ldots, e_n).$ 

A set $S$ of monomials is called \emph{cone-closed} if for every monomial in $S$ all its submonomials are also in $S$. 
\end{definition}

\begin{lemma}[Coef.~extraction]\label{lemma:coef_extraction}
Let $C$ be a circuit which computes an arity-$n$ degree-$d$ polynomial. Then for any monomial $m=\prod_{i\in[n]} x_i^{e_i}$, we have blackbox access to a $\mathrm{poly}(|C|d, \mathrm{cs}(m))$-size circuit computing the coefficient of $m$ in $C$, where $\mathrm{cs}(m)$ denotes the cone-size of $m$.  
\end{lemma}

\begin{proof}
Our proof is in two steps. First, we inductively build a  circuit computing a polynomial which has two parts; one is $\mathrm{coef}_m(C)\cdot m$ and the other one is a ``junk'' polynomial where every monomial is a proper super-monomial of $m$. Second, we construct a circuit which extracts the coefficient of $m$. In both these steps the key is a classic interpolation trick. 

We induct on the variables. For each $i\in[n]$, let $m_{[i]}$ denote $\prod_{j\in[i]} x_j^{e_j}$. We will construct a circuit $C^{(i)}$ which computes a polynomial of the form, 
\begin{equation}\label{eqn-Ci}
C^{(i)}(\mathbf x) \,=\, \mathrm{coef}_{m_{[i]}}(C)\cdot m_{[i]} \,+\, C_{junk}^{(i)} 
\end{equation}
where, for every monomial $m'$ in the support of $C_{junk}^{(i)}$, $m_{[i]}$ is a proper submonomial of $m'_{[i]}$. 

{\em Base case:} Since $C=: C^{(0)}$ computes an arity-$n$ degree-$d$ polynomial, $C(\mathbf x)$ can be written as $C(\mathbf x) \,=\, \sum_{j=0}^d c_jx_1^j$ where, $c_j\in \mathbb{F}[x_2,\ldots, x_n]$. Let $\alpha_0, \ldots, \alpha_{e_1}$ be some
$e_1+1$ distinct elements in $\mathbb F$. For every $\alpha_j$, let $C_{\alpha_jx_1}$ denote the circuit $C(\alpha_jx_1, x_2,\ldots, x_n)$ which computes 
$c_0+c_1\alpha_jx_1+\ldots +c_{e_1}\alpha_j^{e_1}x_1^{e_1}+\cdots+c_d\alpha_j^dx_1^d $ . 
Since $$
M=
\begin{bmatrix}
1      & \alpha_0      & \ldots &  \alpha_0^{e_1}\\
\vdots & \vdots        & \vdots &  \vdots \\
1      &  \alpha_{e_1} & \ldots &\alpha_{e_1}^{e_1}
\end{bmatrix}
$$
is an invertible Vandermonde matrix, one can find an $\mathbf{a}=[a_0,\ldots, a_{e_1}]\in\mathbb F^{e_1+1}$, $\mathbf a\cdot M=[\ 0,\ 0,\ \ldots ,\  1]$ . Using this $\mathbf{a}$, we get the circuit $C^{(1)} := \sum_{j=0}^{e_1}a_jC^{(0)}_{\alpha_jx_1}$ . Its least monomial wrt $x_1$ has $\deg_{x_1}\ge e_1$, which is the property that we wanted.

{\em Induction step $(i\rightarrow i+1)$:} From induction hypothesis, we have the circuit $C^{(i)}$ with the properties mentioned in Eqn.\ref{eqn-Ci}. The polynomial can also be written as $b_0 + b_1x_{i+1}+\ldots+b_{e_{i+1}}x_{i+1}^{e_{i+1}}+\ldots b_dx_{i+1}^d$ , where every $b_j$ is in $\mathbb F[x_1,\ldots, x_{i}, x_{i+2}, \ldots, x_n]$. Like the proof of the base case, for $e_{i+1}+1$ distinct elements $\alpha_0,\ldots, \alpha_{e_{i+1}}\in \mathbb F$, we get 
$ C^{(i+1)} \,=\, \sum_{j=0}^{e_{i+1}} a_j C^{(i)}_{\alpha_jx_{i+1}}$, for some $\mathbf a=[a_0, \ldots, a_{e_{i+1}}]\in \mathbb F^{e_{i+1}+1}$ and  the structural constraint of $C^{(i+1)}$ is easy to verify, completing the induction.

Now we describe the second step of the proof. After first step, we get 
$$C^{(n)}(\mathbf x) \,=\, \mathrm{coef}_{m}(C)\cdot m \,+\, C_{junk}^{(n)} \,,$$ 
where for every monomial $m'$ in the support of $C_{junk}^{(n)}$ , $m$ 
is a proper submonomial of $m'$. Consider the polynomial $C^{(n)}(x_1t, \ldots, x_nt)$ for a fresh variable $t$. Then, using interpolation wrt $t$ we can construct a $O(|C^{(n)}|\cdot d)$-size circuit for $\mathrm{coef}_m(C)\cdot m$, by extracting the coefficient of $t^{\deg(m)}$, since the degree of every monomial appearing in $C^{(n)}_{junk}$ is $>\deg(m)$. Now evaluating at $\mathbf 1$, we get $\mathrm{coef}_m(C)$. The size, or time, constraint of the final circuit clearly depends polynomially on $|C|, d$ and $\mathrm{cs}(m)$.
\end{proof}

But, how many low-cone monomials can there be? Fortunately, in the log-arity regime they are not too many \cite{S13}. Though, in general, they are quasipolynomially many.

\begin{lemma}[Counting low-cones]
\label{lemma:cs_monomials}
The number of arity-$n$ monomials with cone-size at most $k$ is $O(sk^2)$, where 
$s := {\left( 3n/\log k\right)}^{\log k}.$ 
\end{lemma}

\begin{proof}
First, we prove that for any fixed support set, the number of 
cone-size $\leq k$ monomials is less than $k^2$. Next, we multiply 
by the number of possible support sets to get the estimate.
 
Let $T(k,\ell)$ denote the number of cone-size$\leq k$ monomials $m$ with support 
set, say, exactly $\{x_1,\ldots,x_\ell\}$. Since the exponent of $x_{\ell}$ in such an $m$ is at least $1$ and at most $k-1$, we have the following by the disjoint-sum rule: $T(k,\ell)\leq$ $\sum_{i=2}^{k}T\left( k/i, \ell-1 \right)$. This recurrence affords an easy inductive proof as, $T(k,\ell) <$ $\sum_{i=2}^k (k/i)^2 <$ $k^2\cdot \sum_{i=2}^k \left(\frac{1}{i-1}-\frac{1}{i}\right)$ $< k^2$.

From the definition of cone, a cone-size $\leq k$ monomial can have support size at most $\ell:=\floor{\log k}$. The number of possible support sets, thus, is 
$\sum_{i=0}^{\ell} \binom{n}{i}$. Using the binomial estimates \cite[Chap.1]{Jukna}, we get $\sum_{i=0}^{\ell} \binom{n}{i}\leq$ $\left(3n/\ell\right)^\ell$. 
\end{proof}

The partial derivative space of arithmetic circuits has been defined, and mined, in various works \cite{CKW11}. Even when this space is small we do not have efficient hitting-sets known (though \cite{FSS14} gave an $s^{O(\log\log s)}$-time hitting-set.). Below we give a poly-time solution in the log-arity regime.
(It requires $\ch\;\F=0$ or large.)

\begin{proof}[Proof of Thm.\ref{thm-main2}]
The proof has two steps. First, we show that with respect to any 
monomial ordering $\prec$, for all nonzero $C\in \mathcal C$, the dimension of the partial derivative space of $C$ is lower bounded by the cone-size of the leading monomial (that nontrivially occurs) in $C$. Using this, we can get a blackbox PIT algorithm for $\mathcal C$ by testing the coefficients of all the monomials of $C$ of cone-size $\leq k$ for zeroness. Next, we estimate the time complexity to do this. 

The first part is the same as the proof of \cite[Cor.8.4.14]{Forbes} (with origins in \cite{FS13}). Here, we give a brief outline. Let $LM(\cdot)$ be the {\em leading monomial} operator wrt the monomial ordering $\prec$. It can be shown that for any polynomial $f(\mathbf x)$, the dimension of its partial derivative space $\partial_{\mathbf x^{<\infty}}(f)$ is the same as  
$D \,:=\, \#\left\lbrace LM(g) \,\mid\, g\in \partial_{\mathbf x^{<\infty}}(f)
\right\rbrace$ (see \cite[Lem.8.4.12]{Forbes}).
This means that $\dim \partial_{\mathbf x^{<\infty}}(f)$ is lower-bounded by the cone-size of $LM(f)$ \cite[Cor.8.4.13]{Forbes}, which completes the proof 
of our first part. 

Next, we apply Lem.\ref{lemma:coef_extraction}, on the circuit $C$ and a monomial $m$ of cone-size $\leq k$, to get the coefficient of $m$ in $C$ in $\mathrm{poly}(sdk)$-time. Finally, Lem.\ref{lemma:cs_monomials} tells that we have to access at most $k^2\cdot \left(3n/\log k\right)^{\log k}$ many monomials $m$. Multiplying these two expressions gives us the time bound. 
\end{proof}

This gives us immediately,

\begin{corollary}\label{cor:polyPIT_using_cs_for_tinyckt}
Let $\mathcal C$ be a set of arity-$n$ degree-$d$ size-$s$ circuits with 
$n=O(\log sd)$. Suppose that, for all $C\in \mathcal C$, the dimension of the partial derivative space of $C$ is poly$(sd)$. Then, the blackbox PIT for $\mathcal C$ can be solved in poly$(sd)$-time.
\end{corollary}

A \emph{depth-$3$ diagonal circuit} \cite{Sax08} is of the form $C(\mathbf x)=$ $\sum_{i\in[k]} c_i \ell_i^{d_i}$, where $\ell_i$'s are linear polynomials over $\mathbb F$ and $c_i$'s in $\F$. We use $\rk(C)$ to denote the linear rank of the polynomials $\{\ell_i\}_i$.

\begin{theorem}\label{thm:polyPIT_for_tiny_depth-3_diag_ckt}
Let $\mathcal C$ be the set of all arity-$n$ degree-$d$ size-$s$ depth-$3$ diagonal circuits. Suppose that, for all $C\in \mathcal C$, $\rk(C)=O(\log sd)$. Then, the blackbox PIT for $\mathcal C$ can be solved in poly$(sd)$-time.

{\em (Proved in Sec.\ref{app-1})}
\end{theorem}

\vspace{-1mm}
\section{Cone-closed basis after shifting-- Proof of Thm.\ref{thm-main3}}\label{sec-ccb}
\vspace{-1mm}

In this section we will consider polynomials over a vector space, say $\F^k$. This viewpoint has been useful in studying arithmetic branching programs (ABP), eg.~\cite{Agrawal13, FSS14, AGKS15, GKST16}. Let $D\in \F^k[\mathbf x]$ and let $\lrsp(D)$ be the span of its coefficients. We say that $D$ has a {\em cone-closed basis} if there is a cone-closed set of monomials $B$ whose coefficients in $D$ form a basis of $\lrsp(D)$.

This definition is motivated by the fact that there are some models which have this property naturally, for eg.~see Lem.\ref{lemma:cone-closed_basis_for_diagonal_circuit}. In general, this concept subsumes some of the well-known notions of {\em rank concentration} \cite{Agrawal13, FSS14, F15, Forbes}, i.e.~ensuring a basis of $\lrsp(D)$ in a set of monomials that have a small measure in some sense (eg.~cone-size or support-size.). 

\begin{lemma}\label{lemma:cone-closed_cs}
Let $D(\mathbf x)$ be a polynomial in $\mathbb F^k[\mathbf x]$. Suppose that $D(\mathbf x)$ has a cone-closed basis. Then, $D(\mathbf x)$ has $(k+1)$-cone concentration and $(\lg 2k)$-support concentration. 
\end{lemma}
\begin{proof}
Let $B$ be a cone-closed set of monomials forming the basis of $\lrsp(D)$. Clearly, $|B|\le k$. Thus, each $m\in B$ has cone-size $\le k$. In other words, $D$ is $(k+1)$-cone concentrated.

Moreover, each $m\in B$ has support-size $\le\lg k$. In other words, $D$ is $(\lg 2k)$-support concentrated.
\end{proof}

Ideally, we would want to modify a given tiny circuit to get a cone-closed basis. This would solve the PIT problem as shown in the previous section. What are the possible ways to get this? We will show that the concept of basis isolating weight assignment, introduced in \cite{AGKS15}, leads to a cone-closed basis.

\smallskip\noindent {\bf Basis \& weights.}
Consider a weight assignment $\mathbf w$ on the variables $\mathbf x$. It extends to monomials $m=\mathbf{x^e}$ as $\mathbf w(m):= \langle \mathbf {e, w} \rangle = \sum_{i=1}^n e_iw_i$. Sometimes, we also use $\mathbf w(\mathbf e)$ to 
denote $\mathbf w(m)$. Similarly, for a set of monomials $B$, the weight of $B$ is $\mathbf w(B) :=$ $\sum_{m\in B}\mathbf w(m)$. 

Let $B=\{m_1,\ldots, m_{\ell}\}$ resp.~$B'=\{m_1',\ldots, m_{\ell}'\}$ be an ordered set of monomials (non-decreasing wrt $\mathbf w$) that forms a basis of the span of coefficients of $f\in\mathbb F^k[\mathbf x]$. Wrt $\mathbf w$, we say that $B< B'$ if there exists $i\in[\ell]$ such that $\forall j<i,\ \mathbf w(m_j)=\mathbf w(m_j')$ but $\mathbf w(m_i)<\mathbf w(m_i')$. We say that $B\leq B'$ if either $B< B'$ or if $\forall i\in[\ell],\ \mathbf w(m_i)\leq\mathbf w(m_i')$. A basis $B$ is called a \emph{least basis}, if for any other basis $B'$, $B\leq B'$. When is it unique?

A weight assignment $\mathbf w$ is called a \emph{basis isolating weight assignment} for a polynomial $f(\mathbf x)\in\mathbb F^k[\mathbf x]$ if there exists a basis $B$ such that: 
\begin{enumerate}
\item weights of all monomials in $B$ are distinct, and
\item the coefficient of every $m\in\mathrm{supp}(f)\setminus B$ is in the linear span of $\{ \text{coef}_{m'}(f) \;\vert\; m'\in B$, $\mathbf w(m') <$ $\mathbf w(m) \}$.
\end{enumerate}

\begin{lemma}\label{lem-uniq-B}
If $\mathbf w$ is a basis isolating weight assignment for $f$, then $f$ has a unique least basis $B$ wrt $\mathbf w$. 
In particular, for any other basis $B'$ of $f$, we have $\mathbf w(B)<\mathbf w(B')$.
\end{lemma}

\begin{proof}
Let $\ell$ be the dimension of $\lrsp(f)$. Since $\mathbf w$ is a basis isolating weight assignment, we get a basis $B$ that  satisfies the two conditions in the definition of $\mathbf w$. We will show that $B$ is the unique least basis. Let $B=\{m_1,\ldots, m_{\ell}\}$ with $\mathbf w(m_1)<\ldots<\mathbf w(m_{\ell})$. 

Consider any other basis $B'=\{m_1',\ldots, m_{\ell}'\},$ with $\mathbf w(m_1')\leq\ldots\leq\mathbf w(m_{\ell}')$. Let $j$ be the minimum number such that $m_j\neq m_j'$ (it exists as $B\neq B'$). Suppose $\mathbf w(m_j)\geq\mathbf w(m_j')$. 
Since $m_j'\notin B$, the coefficient of $m_j'$ can be written as a linear combination of the coefficients of $m_i$'s for $i<j$. From the definition of $j$, for all $i<j$, $m_i=m_i'$. So the coefficient of $m_j'$ can also be written as a linear combination of the coefficients of $m_i'$'s for $i<j$. This contradicts that $B'$ is a basis and proves that $\mathbf w(m_j) < \mathbf w(m_j')$. 

Now we move beyond $j$. First, we prove that for all $i\in[\ell]$, $\mathbf w(m_i)\leq \mathbf w(m_i')$. For the sake of contradiction assume that there exists a number $a$ such that $\mathbf w(m_a)>\mathbf w(m_a')$. Pick the least such $a$. Let $V$ be the span of the coefficients of monomials in $f$ whose weights are $\leq \mathbf w(m_a')$. Since, for all $i\in[a]$, the coefficient of $m_i'$ is in $V$ and all of them are linearly independent, we know that $\dim(V)\geq a$. On the other hand, for every monomial $m$ in $f$ of $\mathbf w(m)\leq \mathbf w(m_a')<\mathbf w(m_a)$, the coefficient of $m$ can be written as a linear combination of the coefficients of $m_i$'s where $i<a$. This implies that $\dim(V)<a$, which yields a contradiction. Thus, for all $i\in[\ell]$, $\mathbf w(m_i)\leq \mathbf w(m_i')$. In other words, $B\le B'$.

Togetherwith $\mathbf w(m_j) < \mathbf w(m_j')$, we get that $B<B'$ and $\mathbf w(B)<\mathbf w(B')$.
\end{proof}

Next we want to study the effect of shifting $f$ by a basis isolating weight assignment. To do that we require an elaborate notation. As before $f(\mathbf x)$ 
is an arity-$n$ degree-$d$ polynomial over $\mathbb F^k$. For a weight assignment $\mathbf w$, by $f(\mathbf x+t^{\mathbf w})$ we denote the polynomial $f(x_1+t^{w_1}, \ldots, x_n+t^{w_n})$. Let $M=\{\mathbf a\in \mathbb N^n : |\mathbf a|_1\leq d\}$ correspond to the relevant monomials. For every $\mathbf a\in M$, $\cf_{\mathbf{x^a}}(f(\mathbf x+t^{\mathbf w}) )$ can be expanded using the binomial expansion, and we get:
\begin{equation}\label{eqn-binom}
\sum_{\mathbf b\in M}\binom{\mathbf b}{\mathbf a} \cdot t^{\mathbf w(\mathbf b)-\mathbf w(\mathbf a)} \cdot \mathrm{coef}_{\mathbf{x^b}}(f(\mathbf x)) \,.
\end{equation}

We express this data in matrix form as $F'=D^{-1}TD\cdot F$, where the matrices involved are,
\begin{enumerate}
\item $F$ and $F'$: rows are indexed by the elements of $M$ and columns 
are indexed by $[k]$. In $F$ resp.~$F'$ the $\mathbf a$-th row is $\mathrm{coef}_{\mathbf{x^a}}(f(\mathbf x))$ resp.~$\mathrm{coef}_{\mathbf{x^a}}(f(\mathbf x+t^{\mathbf w}))$.
\item $D$: is a diagonal matrix with both the rows and columns indexed by $M$. For $\mathbf a\in M$, $D_{\mathbf a,\mathbf a} := t^{\mathbf w(\mathbf{x^a})}$ .
\item $T$: both the rows and columns are indexed by $M$. For $\mathbf a,\mathbf b\in M$, $T_{\mathbf a,\mathbf b} := \binom{\mathbf b}{\mathbf a}$ .
\end{enumerate}

We will prove the following combinatorial property of $T$: For any $B\subseteq M$, there is a {\em cone-closed} $A\subseteq M$ such that the submatrix $T_{A,B}$ has full rank. Our proof is an involved double-induction, so we describe the construction of $A$ as Algorithm \ref{algo:find_con-closed_set}.
 
\begin{algorithm}
\caption{Finding cone-closed set}
\label{algo:find_con-closed_set}

\begin{algorithmic}
\State \textbf{Input:} A subset $B$ of the $n$-tuples $M$.
\State \textbf{Output:} A cone-closed $A\subseteq M$ with full rank $T_{A,B}$.
 
\Function{Find-Cone-closed}{$B$, $n$}

\If{$n=1$}

  \State $s\leftarrow |B|$;
  
  \Return $\{0. \ldots, s-1\}$;

\Else

  \State Let $\pi_n$ be the map which projects the set of monomials $B$ on the first $n-1$ variables;
  
  \State Let $\ell$ be the maximum number of preimages under $\pi_n$;
  
  \State $\forall i\in[\ell]$, $F_i$ collects those elements in $\mathrm{Img}(\pi_n)$ whose preimage size$\geq i$;
  
  \State $A_0\leftarrow \emptyset$;
  
\For{$i\leftarrow 1$ to $\ell$}

  \State  $S_i\leftarrow$ \textsc{Find-Cone-closed}$(F_i, n-1)$;
  \State  $A_i\leftarrow A_{i-1}\bigcup S_i\times \{i-1\}$;

\EndFor

\Return $A$;

\EndIf

\EndFunction

\end{algorithmic}

\end{algorithm}

\begin{lemma}[Comparison]\label{lem-algo1}
Let $B$ and $B'$ be two nonempty subsets of $M$ such that $B\subseteq B'$. Let 
$A=$ \textsc{Find-Cone-closed}$(B, n)$ and $A'= $ 
\textsc{Find-Cone-closed}$(B', n)$ in Algo.\ref{algo:find_con-closed_set}. 
Then $A\subseteq A'$. Moreover, $|A|=|B|$.
\end{lemma}

\begin{proof}
We prove the lemma using induction on $n$.

{\em Base case $(n=1)$:} For $n=1$, the set $A$ is $\{0,\ldots,|B|-1\}$ and 
the other one $A'$ is $\{0,\ldots, |B'|-1\}$. Since $B$ is a subset of 
$B'$, $|B|\leq |B'|.$ So $A$ is also a subset of $A'$. 

{\em Induction step $(n-1\rightarrow n)$:} Let $\ell$ resp.~$\ell'$ be the bounds on the size of preimages of $\pi_n$ in $B$ resp.~$B'$. To denote the set of all elements in $\mathrm{Img}(\pi_n)$ whose preimage size $\geq i$, we use $F_i$ resp.~$F'_i$. Since $B\subseteq B'$ we have $\ell\leq \ell'$, and for all $i\in[\ell']$, $F_i\subseteq F_i'$. So from induction hypothesis, $S_i\subseteq S_i'$. Since 
$A=\bigcup_{i=1}^{\ell}S_i\times\{i-1\}$ and $A'=\bigcup_{i=1}^{\ell'}S_i'\times\{i-1\}$, we deduce that $A\subseteq A'$.

Note that $|A|=|B|$ is true when $n=1$. Let us prove the induction step from $n-1$ to $n$. Since $|A| = \sum_{i\in[\ell]} |S_i|$, and by induction hypothesis $|S_i|=|F_i|$, we deduce that $|A| = \sum_{i\in[\ell]} |F_i|$. From the definition of $F_i$'s we get that $\mathrm{Img}(\pi_n)=F_1\supseteq F_2\supseteq\cdots\supseteq F_\ell$. A monomial $m\in \pi_n(B)$ that has preimage size $j$, is counted exactly $j$ times in $\sum_{i\in[\ell]} |F_i|$. Thus, $|A|=\sum_{i\in[\ell]} |F_i|=|B|$.
\end{proof}

\begin{lemma}[Closure]
\label{lemma:cone-closed_nvar}
Let $B$ be a nonempty subset of $M$. If $A=$ \textsc{Find-Cone-closed}$(B, n)$ 
in Algo.\ref{algo:find_con-closed_set}, then $A$ is cone-closed. 

{\em (Proved in Sec.\ref{app-2})}
\end{lemma}

We recall a fact that has been used for ROABP PIT. (It requires $\ch\;\F=0$ or large.)

\begin{lemma}\cite[Clm.3.3]{GKS16}
\label{lemma:cone-closed_1var}
Let $a_1, \ldots, a_n$ be distinct non-negative integers. Let $A$ be an $n\times n$ matrix with, $i,j\in[n]$, $A_{i,j} := \binom{a_j}{i-1}$. Then, $A$ is full rank.
\end{lemma}

\begin{lemma}[Full rank]
\label{lemma:cone-closed_nvar_basis}
If $A=$ \textsc{Find-Cone-Closed}$(B, n)$ then $T_{ A, B}$ has full rank. 

{\em (Proved in Sec.\ref{app-2})}
\end{lemma}

Now we are ready to prove our main theorem using the transfer matrix equation.

\begin{proof}[Proof of Thm.\ref{thm-main3}]
As we mentioned in Eqn.\ref{eqn-binom}, the shifted polynomial $f(\mathbf x+t^{\mathbf w})$ yields a matrix equation $F'=D^{-1}TD\cdot F$. Let $k'$ be the rank of $F$. We consider the following two cases. 

{\em Case 1 $(k'<k)$:} We reduce this case to the other one where $k'=k$. Let $S$ be a subset of $k'$ columns such that $F_{M, S}$ has rank $k'$. The matrix $F_{M, S}$ denotes the polynomial $f_S(\mathbf x)\in \mathbb F[\mathbf x]^{k'}$, where $f_S(\mathbf x)$ is the projection of the `vector' $f(\mathbf x)$ on the coordinates indexed by $S$. So, any linear dependence relation among the coefficients of $f(\mathbf x)$  is also valid for $f_S(\mathbf x)$. So $\mathbf w$ is also a basis isolating weight assignment for $f_S(\mathbf x)$. Now from our Case 2, we can claim that $f_S(\mathbf x+t^{\mathbf w})$ has a cone-closed basis $A$. Thus, coefficients of the monomials, corresponding to $A$, in $f(\mathbf x)$ form a basis of $\lrsp(f)$. This implies that $f(\mathbf x+t^{\mathbf w})$ has a cone-closed basis $A$.

\smallskip
{\em Case 2 $(k'=k)$:} Let $B$ be the least basis of $f(\mathbf x)$ wrt $\mathbf w$ and $A=$ \textsc{Find-Cone-closed}$(B, n)$. We prove that the coefficients of monomials in $A$ form a basis of the coefficient space of $f(\mathbf x+t^{\mathbf w})$. To prove this, we show that $\det(F_{A, [k]}')\neq 0$. Define $T' :=TDF$ so that $F'=D^{-1}T'$. Using Cauchy-Binet formula \cite{Z93}, we get that 
$$\det(F_{ A, [k]}') \,=\, \sum_{C\in\binom{M}{k}}\det(D^{-1}_{A, C})\cdot \det(T_{C, [k]}') \,.$$
Since for all $C\in\binom{M}{k}\setminus \{A\}$, the matrix $D^{-1}_{ A, C}$ is singular, we have $\det(F_{A, [k]}') =\det(D^{-1}_{A, A})\cdot \det(T_{A, [k]}')$. Again applying Cauchy-Binet formula for 
$\det(T_{A, [k]}')$, we get 
$$\det(F_{A, [k]}') \,=\, \det(D^{-1}_{A, A})\cdot \sum_{C\in\binom{M}{k}}t^{\mathbf w(C)}\det(T_{ A, C})\cdot \det(F_{C, [k]}) \,.$$ 
From Lem.\ref{lem-uniq-B}, we have that for all basis $C\in\binom{M}{k}\setminus \{B\}$, $\mathbf w(C)> \mathbf w(B)$. The matrix $T_{A, B}$ is nonsingular by Lem.\ref{lemma:cone-closed_nvar_basis}, and the other one $F_{B, [k]}$ is 
nonsingular since $B$ is a basis. Hence, the sum is a nonzero polynomial in $t$. In particular, $\det( F_{ A, [k]}')\ne0$, which ensures that the coefficients of the monomials corresponding to $A$ form a basis of $\lrsp_{\F(t)}(f(\mathbf x+t^{\mathbf w}))$. Since Lem.\ref{lemma:cone-closed_nvar} says that $A$ is also cone-closed, we get that $f(\mathbf x+t^{\mathbf w})$ has a cone-closed basis.
\end{proof}

\vspace{-1mm}
\section{Conclusion}
\vspace{-1mm}

We introduce the tiny diagonal depth-$4$ (resp.~tiny variants of depth-$3$, width-$2$ ABP and extremely low-arity $\Sigma\Pi\Sigma\wedge$ or $\Sigma\wedge^a\Sigma\Pi$) model with the motivation that its poly-time hitting-set would: (1) solve VP PIT (in quasipoly-time) via a poly-time arity reduction ($n\mapsto\log sd$), and (2) prove that either E$\not\subseteq$\#P/poly or VNP has polynomials of arithmetic complexity $2^{\Omega(n)}$. Since now we could focus solely on the PIT of {\em log}-arity VP circuits, we initiate a study of properties that are useful in that regime. These are low-cone concentration and cone-closed basis. Using these concepts we solve a special case of diagonal depth-$3$ circuits. This work throws up a host of tantalizing models and poses several interesting questions: 

\smallskip\noindent
Could the arity reduction phenomenon in Thm.\ref{thm-main1.1} be improved (say, to $2^{2^n}$)? 

\smallskip\noindent
Could we show that the $g$ in Lem.\ref{lemma:PIT_to_lb} is in VNP and not merely E-computable? This would strongly relate PIT to VNP$\ne$VP.

\smallskip\noindent
Could we prove nontrivial lower bounds against the tiny models? 

\smallskip\noindent
Could we solve PIT for size-$s$ $\Sigma\Pi\Sigma\wedge(n)$ in poly($s, \mu(n)$)-time, for some function $\mu$?

\smallskip\noindent
Could we solve PIT for size-$s$ semantic individual-degree-$a'$ $\Sigma\Pi\Sigma(n)$ circuits in poly($s2^n$, $\mu'(a')$)-time, for some function $\mu'$?

\smallskip\noindent
Could we solve PIT for size-$s$ $\Sigma\Pi\Sigma(n)$ in poly($s, \mu(n)$)-time, for some function $\mu$?

\smallskip\noindent
Could we do blackbox PIT for ROABP when $n=\omega(1)$? For instance, given oracle $C=\sum_{i\in[k]}\prod_{j\in[n]} f_{i,j}(x_j)$ of size$\le s$, we want a hitting-set in poly($s,\mu(n)$)-time, for some function $\mu$. It is known that diagonal depth-$3$ blackbox PIT reduces to this problem if we demand $\mu(n)=2^{O(n)}$ \cite{FSS14}.

\smallskip\noindent
Could we do blackbox PIT for size-$s$ arity-($\log s$) individual-degree-($\log\log s$)   ROABPs?

\smallskip\noindent
Could we do blackbox PIT for size-$s$ arity $\omega(\log s)$ multilinear ROABPs?
\vspace{-1mm}
\section*{Acknowledgements}
\vspace{-1mm}

We thank Ramprasad Saptharishi for many useful discussions. M.F.~\& N.S.~thank the organizers of algebraic complexity workshops in 2014 (MPI Saarbr\"ucken \& TIFR Mumbai) that initiated the early discussions. N.S.~thanks the funding support from DST (DST/SJF/MSA-01/2013-14).


\bibliographystyle{alpha}
\bibliography{references}


\appendix

\section{Proofs from Sec.\ref{sec-tdd4}: Other tiny models }\label{app-0}

If a circuit $C$ computes a polynomial of individual-degree $\le d_0$ then we say that the {\em semantic individual-degree} bound of $C$ is $d_0$. Recall the definition of the tiny depth-$3$ model (Defn.\ref{def:tiny_depth-4_diagonal_circuit}).

\begin{reptheorem}{thm-main1.01}[Tiny depth-$3$]
If we have poly-time hitting-sets for a tiny depth-$3$ model, then for VP circuits we have a poly($sd$)-time arity reduction ($n\mapsto O(\log sd)$) that preserves nonzeroness (and proves an exponential lower bound).
\end{reptheorem}
\begin{proof}
The proof strategy is identical to that of Thm.\ref{thm-main1}. So, we will only sketch the main points here.

Let constant $c>1$, and function $\mu'$ be given in the hypothesis. 
Let $\mathcal C\subset \calT'_{\mu',c}$ be the set of tiny depth-$3$ circuits of size $\le s$ and arity $m:=\log s$. Assume that $\mathcal C$ has a $(s^e, s^e)$-hsg $\mathbf f(y)$ for some constant $e\ge1$. Then using Lem.\ref{lemma:PIT_to_lb}, we have an $m$-variate polynomial $q_m$ with individual-degree less than some constant $2\delta$, $\delta^{m-1}=\delta^{\log s-1}> s^e \log s$, and $q_m$ is computable in $s^{O(1)}= \delta^{O(m)}$ time. It has degree $=\delta m$. Importantly, $q_m\notin \calT'_{\mu',c}$ , thus, no tiny depth-$3$ circuit of size $\leq s = \delta^{\Theta(m)}$ can compute it (otherwise, $q_m(\mathbf f(y))\ne 0$ which contradicts its definition as an annihilator). Next we show that it is also not computable by any $\delta^{o(m)}$-size arithmetic circuit. 

For the sake of contradiction, assume that $q_m$ has a $\delta^{o(m)}$-size circuit. Repeat the depth-reduction arguments, as in the proof of Thm.\ref{thm-main1}. Let $a:=5^t$ and $b:=\delta m/2^{t}$. Note that we can ensure $a,b=o(\log s)$, $a=\omega(1)$, and we have a shallow circuit for $q_m$ of the form $\Sigma\Pi^{a}\Sigma\Pi^{b}$. 

It was shown in \cite{Gupta13} that any size-$s'$ $\Sigma\Pi^a\Sigma\Pi^b(n)$ circuit can be transformed to a poly($s'2^{a+b}$)-size $\Sigma\Pi\Sigma^b(n)$ circuit. Applying it here, we get a depth-$3$ circuit $C'$, computing $q_m$, of the form $\Sigma\Pi\Sigma$ and size $\delta^{o(m)}\cdot 2^{a+b} = \delta^{o(m)}$. Moreover, the measure $2^m+\mu'(2\delta)=s+O(1)=o(s^c)$. Thus, $C'$ is a tiny depth-$3$ circuit in $\calT'_{\mu',c}$ , of size $\delta^{o(m)}$ which is $<s$. This contradicts the hardness of $q_m$. Thus, there is no arithmetic circuit for $q_m$ of size $\delta^{o(m)}$. 

Now invoking Lem.\ref{lemma:lb_to_PIT} on the hard family $\{q_m\}_{m\ge1}$, we get our claim. 
Moreover, by Cor.\ref{cor-hs-hard} $\{q_m\}_m$ is an E-computable polynomial family that has arithmetic circuit complexity $2^{\Omega(m)}$.
\end{proof}

We will now consider the polynomials that can be computed by upper-triangular width-$2$ arithmetic branching programs (ABP).

\begin{reptheorem}{thm-main1.02}[Width-$2$ ABP]
If we have poly($s2^n$)-time hitting-sets for size-$s$ arity-$n$ width-$2$ upper-triangular ABP, then for VP circuits we have a poly($sd$)-time arity reduction ($n\mapsto O(\log sd)$) that preserves nonzeroness (and proves an exponential lower bound).
\end{reptheorem}
\begin{proof}
In \cite[Thm.3]{SSS09} an efficient transformation was given that rewrites a size-$s$ arity-$n$ depth-$3$ circuit, times a special product of a linear polynomials, as a poly($s$)-size arity-$n$ width-$2$ upper-triangular ABP. Thus, a poly($s2^n$)-time hitting-set for the latter model gives a poly($s2^n$)-time hitting-set for the former. This by Thm.\ref{thm-main1.01} gives us the exponentially hard polynomial family $\{q_m\}_{m\ge1}$ that is E-computable. Now invoking Lem.\ref{lemma:lb_to_PIT} on the hard family, we get the arity reduction. 
\end{proof}

\begin{reptheorem}{thm-main1.03}[Multilinear tiny depth-$3$]
If we have poly-time hitting-sets for a multilinear tiny depth-$3$ model, then for VP circuits we have a poly($sd$)-time arity reduction ($n\mapsto O(\log sd)$) that preserves nonzeroness (and proves an exponential lower bound).
\end{reptheorem}
\begin{proof}

The proof strategy is identical to that of Thm.\ref{thm-main1}. So, we will only sketch the main points here. Let constant $c>1$, and function $\mu'$ be given in the hypothesis. 
Let $a'= a'(s)$ be a `small' {\em unbounded} function determined by the constraint $\mu'(a')\le s$. Assume that size-$s$ arity-$(a'\log s)$ polynomials in $\calM_{\mu',c}$ have a $(s^e, s^e)$-hsg $\mathbf f(y)$ for some constant $e\ge1$.

Let $\mathcal C\subset \calM_{\mu',c}$ be the set of multilinear tiny depth-$3$ circuits of size $\le s$ and arity $m:=(e+2)\log s < a'\log s$. Since $\mathcal C$ has the $(s^e, s^e)$-hsg $\mathbf f(y)$, so using Lem.\ref{lemma:PIT_to_lb}, we have an arity $m':=(e+1)\log s$ multilinear annihilating polynomial $q'_{m'}(x_{1},\ldots, x_{m'})$ , which is computable in $s^{O(1)}$ time (note: $2^{m'}=2^{(e+1)\log s}=s^{e+1}> s^e\cdot m'= s^e\cdot(e+1)\log s$). Consider its multiple $q_m := q'_{m'}\cdot (x_{m'+1}\cdots x_m)$ which is also multilinear and an annihilator. 
It has degree $\ge (m-m')= \log s$. Importantly, $q_m\notin \calM_{\mu',c}$ , thus, no multilinear tiny depth-$3$ circuit of size $\leq s = 2^{\Theta(m)}$ can compute it (otherwise, $q_m(\mathbf f(y))\ne 0$ which contradicts its definition as an annihilator). Next we show that it is also not computable by any $2^{o(m)}$-size arithmetic circuit. 

For the sake of contradiction, assume that $q_m$ has a $2^{o(m)}$-size circuit. Repeat the depth-reduction arguments, as in the proof of Thm.\ref{thm-main1}, after cutting at depth $t=\omega(1)$. Let $a:=5^t$ and $b:=m/2^{t}$. Note that we can ensure $a,b=o(m)=o(\log s)$, $a=\omega(1)$, and we have a shallow circuit for $q_m$ of the form $\Sigma\Pi^{a}\Sigma\Pi^{b}$. 

It was shown in \cite{Gupta13} that any size-$s'$ $\Sigma\Pi^a\Sigma\Pi^b(n)$ circuit can be transformed to a poly($s'2^{a+b}$)-size $\Sigma\Pi\Sigma^b(n)$ circuit. Applying it here, we get a depth-$3$ circuit $C'$, computing $q_m$, of the form $\Sigma\Pi\Sigma$ and size $2^{o(m)}\cdot 2^{a+b} = 2^{o(m)}$. Moreover, the measure $\mu'(m/\log s)= \mu'(e+2)= O(1)= o(s^c)$. Thus, $C'$ is a multilinear tiny depth-$3$ circuit in $\calM_{\mu',c}$ , of size $2^{o(m)}$ which is $<s$. This contradicts the hardness of $q_m$. Thus, there is no arithmetic circuit for $q_m$ of size $2^{o(m)}$. 

Now invoking Lem.\ref{lemma:lb_to_PIT} on the hard family $\{q_m\}_{m\ge1}$, we get our claim. 
Moreover, by Cor.\ref{cor-hs-hard} $\{q_m\}_m$ is an E-computable polynomial family that has arithmetic circuit complexity $2^{\Omega(m)}$.
\end{proof}

\begin{reptheorem}{thm-main1.3}[Tinier $n,a$]
Fix functions $\mu=\mu(s)$ and $\mu'$.
If we have poly($s, \mu(a), \mu'(n)$)-time hitting-set for size-$s$ $\Sigma\wedge^a\Sigma\Pi(n)$ circuits, then for VP circuits we have a poly($sd$)-time arity reduction ($n\mapsto O(\log sd)$) that preserves nonzeroness (and proves an exponential lower bound).
\end{reptheorem}
\begin{proof}
Suppose we have a poly($s,\mu(a), \mu'(n)$)-time hitting-set $\calH_{s,a,n}$ for size-$s$ $\Sigma\wedge^a\Sigma\Pi(n)$ circuits. Wlog we can assume that $\mu'(n)= \Omega(n)$. Let $a'=a'(s)=\omega(1)$ be a function satisfying $\mu'(a')\le s$, and define $n'=n'(s):= a'\log s$. Consider a size-$s$ $\Sigma\wedge^a\Sigma\Pi(n')$ circuit $C\ne0$ computing a multilinear polynomial. We intend to design a hitting-set for $C$.

Partition the variable set $\{x_1,\ldots,x_{n'}\}$ into $a'$ blocks $B_j, j\in[a']$, each of size $\log s$. Recall the map $\varphi$ designed in the proof of Thm.\ref{thm-main1.2}. We have $\varphi(C)\ne0$ (basically, use the fact that $C$ computes a nonzero multilinear polynomial and $\varphi$ keeps the multilinear monomials distinct). Finally, $\varphi(C)$ becomes an arity-$a'$ $\Sigma\wedge^a\Sigma\Pi$ circuit of size at most $s + s\cdot 2^{\log s}= O(s^2) $. 
Thus, using $\calH_{O(s^2), a, a'}$ we get a hitting-set for $\varphi(C)$ of time-complexity poly($s^2,\mu(a), \mu'(a')$)= poly($s,\mu(a)$). 

In turn, we get a poly($s,\mu(a)$)-time hitting-set for $\Sigma\wedge^a\Sigma\Pi$ circuits that compute multilinear polynomials of arity $a'\log s=\omega(\log s)$. Using an argument identical to that in the proof of Thm.\ref{thm-main1.03}, we will get the usual E-computable polynomial family $\{q_m\}_m$ that has no arithmetic circuit of size $2^{o(m)}$. This finishes the proof by invoking Lem.\ref{lemma:lb_to_PIT} and Cor.\ref{cor-hs-hard}.
\end{proof}

\section{Proofs from Sec.\ref{sec-lcc} }\label{app-1}

An arity-$n$ depth-3 diagonal circuit over $\F$ can be written as 
$C(\mathbf x) = \sum_{i=1}^k c_i\tilde f_i^{d_i}$, where $\tilde f_i$'s are linear 
polynomials. Let $f_i$ be the non-constant part of $\tilde f_i$ for all $i\in[k]$. 
Suppose that $\rk_{\F}\{f_1, \ldots f_k\}=:r$. Wlog, we can assume that $f_1,\ldots,f_r$ is a basis of the space spanned by $f_i$'s. Then there exists an $r$-variate polynomial $A(\mathbf z)$ such that $C(\mathbf x)= A(f_1,\ldots, f_r)$. Let $L_{\F[\mathbf x]}$, where $\mathbf x=(x_1,\ldots, x_n)$, resp.~$L_{\F[\mathbf y]}$, where $\mathbf y=(y_1,\ldots, y_r)$, be the vector space of linear polynomials over $\F$. 

Using the construction of \cite[Sec.3.2]{SS12}, in $\mathrm{poly}(knd)$-time, we can find a linear transformation $\Psi:$ $L_{\F[\mathbf x]}\rightarrow$  $L_{\F[\mathbf y]}$  such that  $\rk_{\F}\{\Psi(f_1),\ldots, \Psi(f_r)\}=r$ and $g_i:=\Psi(f_i)$ are linear forms (i.e.~homogeneous and degree one). Now we prove the following fact which will ensure the non-zeroness of $C(\Psi(\mathbf x))$.

\begin{lemma}
If $A(g_1,\ldots, g_r)=0$ then $A$ is the zero polynomial.
\end{lemma}

\begin{proof}
Since $g_1, \ldots, g_r$ are $r$ linearly independent linear forms on $(y_1,\ldots, y_r)$, we have an invertible linear map $\tau$ from $L_{\F[\mathbf y]}$ to itself such that $\tau(g_i)=y_i$, equivalently, $\tau^{-1}(y_i)= g_i$. Thus, $\tau^{-1}$ induces an $\F$-automorphism $\tilde{\tau}$ on $\F[\mathbf y]$. 

Suppose that $A(g_1,\ldots, g_r)=0$. Then, applying $\tilde{\tau}$ on $A(g_1,\ldots, g_r)$, we get $A(y_1,\ldots, y_r)=0$, thus $A(\mathbf z)=0$. 
\end{proof}

\begin{reptheorem}{thm:polyPIT_for_tiny_depth-3_diag_ckt}
Let $\mathcal C$ be the set of all arity-$n$ degree-$d$ size-$s$ depth-$3$ diagonal circuits. Suppose that, for all $C\in \mathcal C$, $\rk(C)=O(\log sd)$. Then, the blackbox PIT for $\mathcal C$ can be solved in poly$(sd)$-time.
\end{reptheorem}

\begin{proof}
The above description gives us a nonzeroness preserving arity reduction ($n\mapsto \rk(C)$) method that reduces $C$ to an $O(\log(sd))$-variate degree-$d$ 
$\mathrm{poly}(s)$-size depth-$3$ diagonal circuit $C'$. 

Clearly the dimension of the partial derivative space of $C'$ is $\mathrm{poly}(sd)$ \cite[Lem.8.4.8]{Forbes}. Hence, Cor.\ref{cor:polyPIT_using_cs_for_tinyckt} gives us a poly($sd$)-time hitting-set for $C'$.
\end{proof}

\section{Proofs from Sec.\ref{sec-ccb} }\label{app-2}

\begin{replemma}{lemma:cone-closed_nvar}
Let $B$ be a nonempty subset of $M$. If $A=$ \textsc{Find-Cone-closed}$(B, n)$ 
in Algo.\ref{algo:find_con-closed_set}, then $A$ is cone-closed.
\end{replemma}

\begin{proof}
We prove it by induction on $n$.
 
{\em Base case $(n=1)$:} For $n=1$, $A= \{0, \ldots, |B|-1\}$. 
So $A$ is cone-closed.

{\em Induction step $(n-1\rightarrow n)$:} Now $A=\bigcup_{i=1}^{\ell}S_i\times \{i-1\}$ . Let $\mathbf f$ be an element in $A$ and $\mathbf{x^e}$ be a submonomial of $\mathbf{x^f}$. We will show that $\mathbf e\in A$. Let $\mathbf f=: (\mathbf f', k)$ and $\mathbf e=: (\mathbf e', t)$, so that $t\le k$. We divide our proof into the following two cases.

{\em Case 1 $(t=k)$:} We have $\mathbf f'\in S_{k+1} =$ 
\textsc{Find-Cone-closed}$(F_{k+1}, n-1)$. By induction hypothesis, $S_{k+1}$ is cone-closed. Since $\mathbf e'\leq \mathbf f'$, we get $\mathbf e'\in S_{k+1}$. So, $\mathbf e=(\mathbf e',k)\in S_{k+1}\times \{k\}$, which implies that it is also in $A$. 

{\em Case 2 $(t<k)$:} We have $F_{k+1}\subseteq F_{t+1}$. By Lem.\ref{lem-algo1}, we get $S_{k+1}\subseteq S_{t+1}$. So $\mathbf f'\in S_{t+1}$. From induction hypothesis, $S_{t+1}$ is a cone-closed set. This implies that $\mathbf e'\in S_{t+1}$ and $\mathbf e\in S_{t+1}\times \{t\}$. Thus, $\mathbf e$ is also in $A$.

Since $\mathbf e$ was arbitrary, we deduce that $A$ is cone-closed.
\end{proof}

\begin{replemma}{lemma:cone-closed_nvar_basis}
If $A=$ \textsc{Find-Cone-Closed}$(B, n)$ then $T_{ A, B}$ has full rank.
\end{replemma}

\begin{proof}
The proof will be by double-induction-- outer induction on $n$ and an inner induction on iteration $i$ of the `for' loop (Algo.\ref{algo:find_con-closed_set}). 

{\em Base case:} For $n=1$, the claim is true due to Lem.\ref{lemma:cone-closed_1var}. 

{\em Induction step $(n-1\rightarrow n)$:} To show $T_{ A, B}$ full rank, we 
prove that for any vector $\mathbf b \in \mathbb{ F}^{|B|}$: if $T_{ A, B}\cdot 
\mathbf b=0$ then $\mathbf b=0$. For this we show that the following 
invariant holds at the end of each iteration $i$ of the `for' loop (Algo.\ref{algo:find_con-closed_set}).

\smallskip
{\em Invariant (arity-$n$ \& $i$-th iteration):} For each $\mathbf f\in B$ such that the preimage size of $\pi_n(\mathbf f)$ is at most $i$, the product $T_{ A_i, B}\cdot \mathbf b=0$ implies that $\mathbf b_{\mathbf f}=0$. 

\smallskip
At the end of iteration $i=1$, we have the vector $T_{ A_1, B}\cdot \mathbf b$. 
Recall that $A_1=S_1\times\{0\}$ and $F_1=\pi_n(B)$. So $T_{ A_1, B}\cdot \mathbf b =$ $T_{ S_1, F_1}\cdot \mathbf c$, where for $\mathbf e\in F_1$, 
$\mathbf c_{\mathbf e} :=$ $\sum_{(\mathbf e, k)\;\in\; \pi_n^{-1}(\mathbf e)}
\binom{k}{0} \mathbf b_{(\mathbf e, k)}$. Thus, $T_{A_1, B}\cdot \mathbf b=0$ implies $T_{S_1, F_1}\cdot \mathbf c=0$. Since $S_1=$ \textsc{Find-Cone-closed}$(F_1, n-1)$, using induction hypothesis, we get that $\mathbf c=0$. This means that 
for $\mathbf e\in B$ such that the preimage size of $\pi_n(\mathbf e)$ is at most $1$, we have $\mathbf c_{\mathbf e}=0$. This proves our invariant at the end of the iteration $i=1$.

$(i-1\rightarrow i)$:
Suppose that at the end of $(i-1)$-th iteration, the invariant holds. We show that it also holds at the end of the $i$-th iteration. For each $j\in[i]$, let $\mathbf v_j$ denote the projection of $T_{ A_i, B}\cdot \mathbf b$ on the coordinates indexed by $S_j\times \{j-1\}$. By focusing on the latter rows of $T_{ A_i, B}$, we can see that $\mathbf v_j = T_{S_j, F_1}\cdot \mathbf c_j$ where the vector $\mathbf c_j$ is defined as, for $\mathbf e\in F_1$,
\begin{equation}
\label{equation:1}
(\mathbf c_{j})_{\mathbf e} \,:=\, \sum_{(\mathbf e, k) \;\in\;\pi_n^{-1}(\mathbf e)} \binom{k}{j-1} \cdot\mathbf b_{(\mathbf e, k)} \,.
\end{equation}
Suppose that $T_{A_i, B}\cdot \mathbf b=0$. All we have to argue is that for every $\mathbf f\in B$ such that the preimage size of $\mathbf e:=\pi_n(\mathbf f)$ is $i$, the coordinate $\mathbf b_{\mathbf f}=0$. 

Since $T_{A_i, B}\cdot \mathbf b=0$, its projection $\mathbf v_j = T_{S_j, F_1}\cdot \mathbf c_j$ is zero too. By induction hypothesis (on $i-1$), for each $\mathbf e\in F_1$ with preimage size $<i$, the coordinate $(\mathbf c_{j})_{\mathbf e}=0$. Thus, the vector $T_{ S_j, F_1}\cdot \mathbf c_j = T_{S_j, F_j}\cdot \mathbf c'_j$ where the vector $\mathbf c'_j$ is defined as, for $\mathbf e\in F_j$, $(\mathbf c'_{j})_{\mathbf e} := \mathbf c_{j\mathbf e}$. Consequently, $T_{S_j, F_j}\cdot \mathbf c'_j =0$, for $j\in[i]$. By induction hypothesis (on $n-1$), we know that $T_{S_j, F_j}$ is full rank. So $\mathbf c'_j=0$, which tells us that $\mathbf c_j=0$, for $j\in[i]$. 

Fix an $\mathbf e\in F_1$, with preimage size $=i$, and let the preimages be $\{(\mathbf e, k_1), \ldots, (\mathbf e, k_i)\}$ where $k_j$'s are distinct nonnegative integers. Since $\mathbf c_j=0$, for $j\in[i]$, we get from Eqn.\ref{equation:1} and Lem.\ref{lemma:cone-closed_1var} that: $\mathbf b_{(\mathbf e, k_j)}=0$ for all $j\in[i]$. In other words, for any $\mathbf f\in B$ such that the preimage size of $\pi_n(\mathbf f)$ is $i$, the coordinate $\mathbf b_{\mathbf f}=0$. 

$(i=\ell)$:
Since $A=A_{\ell}$, the output of \textsc{Find-Cone-closed}$(B, n)$, using our invariant at the end of $\ell$-th iteration we deduce that $T_{A, B}\cdot \mathbf b=0$ implies $\mathbf b=0$. Thus, $T_{A, B}$ has full rank.
\end{proof}

\section{Models with a cone-closed basis}

We give a simple proof showing that a typical diagonal depth-$3$ circuit is already cone-closed. Consider the polynomial $D(\mathbf x)=(\mathbf 1+\mathbf a_1x_1+\ldots+\mathbf a_n x_n)^d $ in $\F^k[\mathbf x]$, where $\F^k$ is seen as an $\F$-algebra with coordinate-wise multiplication. 

\begin{lemma}
\label{lemma:cone-closed_basis_for_diagonal_circuit}
$D(\mathbf x)$ has a cone-closed basis.
\end{lemma}

\begin{proof}
Consider the $n$-tuple $L:=(\mathbf{a}_1,\ldots, \mathbf{a}_n)$. Then for every monomial $\mathbf x^{\mathbf e}$, the coefficient of $\mathbf x^{\mathbf e}$ in $D$ is $L^{\mathbf e} := \prod_{i=1}^n\mathbf{a}_i^{e_i}$, with some nonzero scalar factor (note: here we seem to need $\ch(\F)$ zero or large). We ignore this constant factor, since it does not affect linear dependence relations. Consider any proper monomial ordering $\prec$ (eg.~{\em deg-lex}). Now we prove that the `least basis' of $D(\mathbf x)$ with respect to this monomial ordering is cone-closed. 

We incrementally devise a monomial set $B$ as follows: Arrange all the monomials in ascending order. Starting from least monomial, put a monomial in $B$ if its coefficient can not be written as a linear combination of its previous (thus, smaller) monomials. From construction, the coefficients of monomials in $B$ form 
the least basis for the coefficient space of $D(\mathbf x)$. Now we show that $B$ is cone-closed. We prove it by contradiction.

Let $\mathbf x^{\mathbf f}\in B$ and let $\mathbf x^{\mathbf e}$ be its submonomial that is not in $B$. Then we can write 
$$L^{\mathbf e} \,=\, \sum_{\mathbf x^{\mathbf b} \prec \mathbf x^{\mathbf e}} c_{\mathbf b} L^{\mathbf b} \text{~~with }c_{\mathbf b}\text{'s in }\F\,.$$ 
Multiplying by $L^{\mathbf f-\mathbf e}$ on both sides, we get 
$$L^{\mathbf f} \,=\, \sum_{\mathbf x^{\mathbf b} \prec \mathbf x^{\mathbf e}}c_{\mathbf b} L^{\mathbf{b+f-e}} \,=\, \sum_{\mathbf x^{\mathbf b'} \prec \mathbf x^{\mathbf{f}}} c'_{\mathbf b'} L^{\mathbf b'} \,. $$

Note that $\mathbf x^{\mathbf b'} \prec \mathbf x^{\mathbf f}$ holds true by the way a monomial ordering is defined. This equation contradicts the fact that $\mathbf x^{\mathbf f}\in B$, and completes the proof.
\end{proof}

\end{document}